\documentclass[12pt]{article}

\usepackage{amssymb}
\usepackage{amsmath}
\usepackage{amsthm}
\usepackage{mathrsfs}
\usepackage{MnSymbol}
\usepackage{cite}
\usepackage{appendix}
\usepackage{enumerate}

\newcommand{\rnet}{\{$\mathscr{S,C,R}\}\,$}
\newcommand{\RS}{\ensuremath{\mathbb{R}^{\mathscr{S}}}}
\newcommand{\PS}{\ensuremath{\mathbb{R}_+^{\mathscr{S}}}}
\newcommand{\PR}{\ensuremath{\mathbb{R}_+^{\mathscr{R}}}}
\newcommand{\RC}{\ensuremath{\mathbb{R}^{\mathscr{C}}}}
\newcommand{\RR}{\ensuremath{\mathbb{R}^{\mathscr{R}}}}
\newcommand{\PbarS}{\ensuremath{\overline{\mathbb{R}}_+^{\mathscr{S}}}}
\newcommand{\kinsys}{\{$\mathscr{S,C,R,K}\}$}
\newcommand{\mas}{\{$\mathscr{S,C,R},k\}$}
\newcommand{\scrS}{\ensuremath{\mathscr{S}}}
\newcommand{\scrC}{\ensuremath{\mathscr{C}}}
\newcommand{\scrR}{\ensuremath{\mathscr{R}}}

\newcommand{\scrK}{\ensuremath{\mathscr{K}}}
\newcommand{\scrM}{\ensuremath{\mathscr{M}}}
\newcommand{\scrI}{\ensuremath{\mathscr{I}}}
\newcommand{\rxn}{\ensuremath{y \to y'}}
\newcommand{\supp}{\ensuremath{\mathrm{supp} \,}}
\newcommand{\sgn}{\ensuremath{\mathrm{sgn} \,}}
\newcommand{\Span}{\ensuremath{\mathrm{span} \,}}

\theoremstyle{plain}
\newtheorem{theorem}{Theorem}[section]

\newtheorem{lemma}[theorem]{Lemma}
\newtheorem{proposition}[theorem]{Proposition}
\newtheorem{corollary}[theorem]{Corollary}
\theoremstyle{definition}
\newtheorem{definition}[theorem]{Definition}

\theoremstyle{remark}
\newtheorem{rem}[theorem]{Remark}

\begin{document}

\author{Guy Shinar\thanks{Department of Molecular Cell Biology, Weizmann Institute of Science, Rehovot 76100, Israel. E-mail: shinarg@gmail.com. Present address: 12 Metzada St., Ramat Gan 52235, Israel. GS was supported by an advanced ERC grant. }
\and Martin Feinberg\thanks{The William G. Lowrie Department of Chemical \& Biomolecular Engineering and Department of Mathematics, Ohio State University, 140 W. 19th Avenue, Columbus, OH, USA 43210.  E-mail: feinberg.14@osu.edu. MF was supported by NSF grant EF-1038394 and NIH grant 1R01GM086881-01.} \thanks{Corresponding author. Phone: 614-688-4883.}}

\title{\emph{Concordant Chemical Reaction Networks}}

\date{\emph{\today}}
\maketitle
\begin{abstract}

	We describe a large class of chemical reaction networks, those endowed with a subtle structural property called \emph{concordance}. We show that the class of concordant networks coincides \emph{precisely} with the class of networks which, when taken with any weakly monotonic kinetics, invariably give rise to kinetic systems that are \emph{injective} --- a quality that, among other things, precludes the possibility of switch-like transitions between distinct positive steady states. We also provide persistence characteristics of concordant networks, instability implications of discordance, and consequences of stronger variants of concordance.  Some of our results are in the spirit of recent ones by Banaji and Craciun, but here we do not require that every species suffer a degradation reaction. This is especially important in studying biochemical networks, for which it is rare to have all species degrade.

\end{abstract}
\section{Introduction}
\label{sec:Intro}
 Although chemical reaction networks can be highly intricate and their induced differential equations can be extraordinarily complex, it is typically the case that dynamical behavior resulting from them is tame.

	Surprisingly often, for example, there is \emph{precisely one} stationary state, even when kinetic parameters range over all feasible values.  This is in contrast to the behavior of complex systems of differential equations in general, for which multiple stationary states are common and, in fact, are to be  expected. Viewed against a background of differential equations generally, then, those that derive from chemical reaction networks appear to admit bistable behavior far less often than might normally be supposed. This dullness of dynamics seems extant across very large classes of reaction networks and persists robustly, even against large changes in parameter values within kinetic models for the rates of individual chemical reactions.

	The implications for the design of biological circuitry are especially compelling. Evolution has certainly fashioned special reaction network modules with dynamical behavior sufficiently rich as to underlie, for example, the bistable switching devices or the oscillatory timekeeping mechanisms necessary for the machinery of life. But these special actors probably require a fairly stable and reliable, yet dynamic, chemical background against which they can play their roles. In this sense, biology-in-the-large would seem to require that only exceptional reaction networks have the capacity to exhibit very rich dynamics, while most others execute their functions in steadier, less surprising ways.

	To understand, in a general way, the design principles by which biology operates, it therefore becomes important to understand the perhaps subtle distinctions between reaction network architectures that, on one hand, might underlie specialized dynamically-rich devices and, on the other hand, architectures that, by their very nature, enforce duller, more restricted behavior despite what might be great intricacy in the interplay of many species, even independently of values that kinetic parameters might take.

	This article is about a large class of chemical reaction networks, called \emph{concordant networks}, for which a certain dullness of behavior is enforced for quite general kinetics. (Concordance is a property of the network itself, divorced from any assignment of a kinetics.)  Whether a particular network is concordant can be readily checked, most easily with help of software of the kind made available in \cite{Ji_toolboxWinV21}, and, as we will show in a separate article, it can sometimes be affirmed by inspection of what has come to be known as the network's  \emph{Species-Reaction Graph} \cite{craciun_multiple_2006-1,banaji_graph-theoretic_2010,craciun_understanding_2006-1}. (See Remark \ref{rem:SRGraphConditions}.) 

	Taken with what we call a \emph{weakly monotonic} kinetics --- essentially a kinetics in which an increase in the rate of a particular reaction requires that there be an increase in the concentration of at least one of its reactant species --- a concordant network has the following properties (among others): 

\medskip

(i) Not only can the resulting kinetic system not admit two distinct stoichiometrically compatible positive\footnote{A positive composition is one in which all species concentrations are strictly positive.} equilibria, it cannot even admit two distinct stoichiometrically compatible equilibria, at least one of which is positive. This is to say that \emph{if} the resulting kinetic system admits a positive equilibrium, it cannot admit \emph{any} second equilibrium, \emph{positive or otherwise}, that is stoichiometrically compatible with the first.

\medskip 
(ii) If the network is weakly reversible and conservative, and if the kinetics is continuous, then among all compositions that are stoichiometrically compatible with some positive composition --- including ``boundary compositions" (i.e., those for which at least one species concentration is zero) --- there \emph{will} indeed be \emph{precisely one} equilibrium, and it is positive.  In fact, the boundary of the set of all such compositions is repulsive against approach from a positive composition in the following sense: At every boundary composition, the production rate of each species with concentration zero is non-negative, and there is at least one such species whose production rate is strictly positive. Moreover, in at least some instances --- when the kinetics is smooth and the network's ``fully open extension" is also concordant (Section \ref{sec:FullyOpenSection}) --- we can assert that the unique equilibrium has a degree of stability: every real eigenvalue associated with it is negative. 

\bigskip

	In describing these features of concordant networks, we have only told part of the story that will unfold later on, but it is perhaps already clear that, for concordant networks, the underlying network structure enforces a very circumscribed kind of behavior. No matter how intricate a concordant network might be --- and concordant networks can be highly intricate --- it cannot, for example, provide the basis for a bistable switch involving two equilibria, at least one positive, so long as the kinetics is weakly monotonic.

	This is not to say that concordance in a network precludes \emph{all} kinds of rich behavior over \emph{all possible} choices of weakly monotonic behavior. It does not. (See Remarks \ref{rem:NotAllRichBehaviorExcluded} and \ref{rem:Bifurcations}.) What is surprising, however, is that concordance --- or the lack of it --- tells us so much. 

	The remainder of this article is organized as follows: In Section \ref{sec:Background} we discuss some earlier, related results and, in particular, connections to the striking work of Banaji and Craciun \cite{ banaji_graph-theoretic_2009, banaji_graph-theoretic_2010}. In Section \ref{sec:CRNTPreliminaries} we provide the chemical reaction network theory ideas necessary for the rest of the article. 

	In Section \ref{sec:ConcordantNetworks} we provide the definitions of concordance, weakly monotonic kinetics, and injectivity of the species formation rate function. In the same section we show in a very simple way that  \emph{for a concordant reaction network, taken with weakly monotonic kinetics, the species formation rate function is injective. A direct consequence of this is the impossibility of two distinct stoichiometrically compatible equilibria, at least one of which is positive.} In fact, we show that, for weakly monotonic kinetics, injectivity is, in a certain sense, 
\emph{synonymous} with concordance: \emph{A reaction network has injectivity in all weakly monotonic kinetic systems derived from it if and only if the network is concordant.} With the results of Section \ref{sec:ConcordantNetworks} behind us, we return briefly in Section \ref{sec:Connections}  to a more fully articulated discussion of connections to other work.

	In Section \ref{sec:Persistence} we discuss persistence properties of kinetic systems derived from concordant networks, in particular boundary behavior of the species formation rate function. Most of the results obtained there rely only on very basic, natural features of kinetic rate functions and do not invoke further requirements such as weak monotonicity.  

	Section \ref{sec:FullyOpenSection} is about connections between concordance of a network and concordance of a network's ``fully open extension" --- roughly, the network obtained by requiring that every species $s$ of the original network suffers a degradation reaction $s \to 0$. It is in this section that we begin to make statements about eigenvalues associated with equilibria. In particular, we show how network concordance enforces a certain degree of stability of positive equilibria for \emph{all} (differentiably) monotonic kinetics. By way of contrast, in Section \ref{sec:Discordance} we discuss  consequences of network \emph{discordance} beyond those stated in Section \ref{sec:ConcordantNetworks}, in particular the guarantee of \emph{unstable} positive equilibria for at least certain instances of differentiably monotonic kinetics.  Among other things we show that for every weakly reversible discordant network there is a differentiably monotonic kinetics such that the resulting kinetic system admits a positive unstable equilibrium.

	In Section \ref{sec:generalizations} we discuss some variations of the definition of concordance and their consequences. There we consider very broad kinetics, in which the rate of a particular reaction might even be influenced by the concentration of a species that, for the reaction, is neither a reactant nor a product. As an important special case, we introduce the notion of \emph{strong concordance} and show that, when the kinetics is similar to the kind admitted by Banaji and Craciun \cite{banaji_graph-theoretic_2010},  then strong concordance suffices for ensuring injectivity of the species-formation rate function (and the absence of multiple stoichiometrically compatible equilibria, at least one of which is positive). This is also discussed in a preliminary way in Section \ref{sec:Background}, which we turn to next.

\section{Some background and a look ahead}
\label{sec:Background}

	Chemical reaction network theory has already been able to delineate large classes of networks for which dull behavior is ensured, at least with respect to certain dynamical features (e.g., the absence of multiple positive stoichiometrically compatible steady states), mostly when the individual reaction rates are presumed to be governed by mass action kinetics. In this section we discuss two quite separate strands of reaction network theory, and then, by way of a preview, we briefly discuss the standing of concordance theory in relation to these.

\subsection{Deficiency-oriented theory}
\label{sec:deficiencyresults}
	The earliest results were centered about the classification of reaction networks by means of a non-negative integer index called the \emph{deficiency}. Thus, there are networks of deficiency zero, of deficiency one, and so on. The deficiency is not a measure of a network's size. In fact, a very large complex network can have a deficiency of zero. Nevertheless, the \emph{Deficiency Zero Theorem} indicates the sense in which all deficiency zero networks taken with mass action kinetics give rise to quite dull dynamics, including the absence of multiple stoichiometrically compatible positive equilibria, the absence of periodic positive composition trajectories, and the absence of an unstable positive equilibrium. The \emph{Deficiency One Theorem} describes a large class of mass action networks for which multiple stoichiometrically compatible positive equilibria are precluded, regardless of rate constant values.  Still other parts of deficiency-oriented theory translate questions about a mass action network's capacity for multiple stoichiometrically-compatible positive equilibria --- an apparently nonlinear problem --- into question about systems of \emph{linear inequalities}. For a survey of some deficiency-oriented reaction network theory see \cite{feinberg_chemical_1987,feinberg_chemical_1988,
feinberg_lectureschemical_1979,
feinberg_existence_1995-1,feinberg_multiple_1995,ellison_catalytic_2000}. The deficiency-oriented theory continues to evolve \cite{Ji_Thesis,Ji_toolboxWinV21}.

\subsection{Theory aimed at fully-open reactors}
\label{sec:fullyopen}
	A rather different strain of reaction network theory began with the Ph.D. work of Paul Schlosser \cite{schlosser_theory_1994,Schlosser_PhDThesis} and the advent of what was called the \emph{Species-Complex-Linkage Graph (SCL Graph)}. When a reaction network's SCL Graph satisfies certain conditions, involving its cycles, then the hypothesis of a certain technical proposition is also satisfied, and this proposition, in turn, ensures that multiple positive steady states are impossible. 

	It is important to understand that this work was tailored very specifically to determining when a reaction network, taken with mass action kinetics, has the capacity to admit multiple positive steady states in the context of what chemical engineers call the (isothermal) \emph{continuous flow stirred tank reactor (CFSTR)}, in which all species are present in an effluent stream. In reaction network theory terms, one studies such reactors by adjoining to the network of ``true" chemical reactions additional ``degradation reactions"  of the form $s\to 0$, one for each species $s$, to account for the efflux of species $s$ in the reactor's effluent \cite{feinberg_chemical_1987}. (There are sometimes additional reactions of the form $0 \to s$ to account for infusion of species $s$ in the feed stream.) Thus, work originating in the Schlosser thesis was aimed specifically at understanding the capacity for multiple positive equilibria for \emph{fully-open reactors}, characterized by reaction networks \emph{in which there is a degradation reaction corresponding to each and every species}.

	This work was later broadened in \cite{craciun_multiple_2005-1,craciun_multiple_2006-1,craciun_understanding_2006-1} by Craciun and Feinberg to give results which are, in many respects, more powerful, using different theoretical underpinnings, but again for mass action kinetics and again for reaction networks \emph{in which there is a degradation reaction corresponding to each species}. If the kinetics is mass action, and, for the network at hand, it is the case that the Jacobian of the species-formation-rate function is nonsingular at every positive composition and for every assignment of rate constants, then it was shown that the species-formation-rate function is injective (relative to positive compositions) for every assignment of rate constants. In this case, multiple positive equilibria become impossible. 

	Two techniques were deployed to determine when a network satisfies ``the Jacobian condition." One was analytical \cite{craciun_multiple_2005-1}, relying on a symbolic determinant expansion of the Jacobian with the aim of establishing that, for the network under study, every term in the expansion is non-negative with at least one positive 
(whereupon the Jacobian must be nonsingular). The other \cite{craciun_multiple_2006-1,craciun_understanding_2006-1}, building on the first, was graphical:  When the network's \emph{Species-Reaction Graph (SR Graph)} satisfies  certain mild conditions, closely resembling the earlier ones for the SCL Graph, then the Jacobian condition is satisfied, and the  species-formation-rate function is injective.\
\medskip

\begin{rem} 
\label{rem:SRGraphConditions}
For the benefit of readers already familiar with results based on the Species-Reaction Graph, we reiterate those conditions here:  \emph{Every even cycle of the network's Species-Reaction Graph is an s-cycle and no two even cycles have a species-to-reaction intersection.} In a separate article we intend to show that a ``normal" reaction network (Section \ref{sec:FullyOpenSection}) is concordant --- in fact, strongly concordant --- if these same conditions are satisfied by the network's Species-Reaction Graph.
\end{rem}
\medskip
	In remarkable and highly surprising subsequent work, Banaji and Craciun \cite{banaji_graph-theoretic_2009,banaji_graph-theoretic_2010} were able to extend results in \cite{craciun_multiple_2006-1,craciun_understanding_2006-1} to a \emph{far} broader class of kinetics, \emph{but again for networks in which there is a degradation reaction for every species}. The kinetics studied by Banaji and Craciun --- which they call \emph{nonautocatalytic (NAC) kinetics}   --- is required to have the property that, for a particular reaction,  an increase [decrease] in the concentration of one of its reactant [product] species --- keeping all other concentrations fixed --- cannot result in a decrease of the reaction rate.  These conditions are expressed in an essential way in terms of derivatives, so, for them, the kinetics is required also to have a degree of smoothness.

	The Banaji-Craciun work built upon earlier work by Banaji and co-workers \cite{banaji_PMatrix_2007}, in which it was shown that if the so-called stoichiometric matrix of a reaction network has the ``SSD property"  --- i.e, the property that \emph{all} of its square sub-matrices are either singular or else ``sign-nonsingular" --- and \emph{if there is a degradation reaction for every species},  then, when the kinetics is NAC, the species formation rate function is injective (and multiple equilibria are excluded).

	Banaji and Craciun \cite{banaji_graph-theoretic_2010} then established, surprisingly, that the Species-Reaction Graph properties shown to ensure injectivity for certain networks in the mass-action case also suffice to ensure that the stoichiometric matrix has the ``SSD property," which implies in turn that, \emph{for fully open reactors}, one has injectivity of the species-formation rate function even for the very broad class of NAC kinetics.
	
	The requirement that there be a degradation reaction for \emph{every} species figures heavily in the Banaji-Craciun result and in the earlier work by Banaji and co-workers: It plays an essential role in connecting the ``SSD property" of the stoichiometric matrix, via the Gale-Nikaido theorem \cite{gale_jacobian_1965},  to injectivity of the species-formation rate function. In contrast to the deficiency-oriented results described in the preceding subsection, this strong reliance on the fully-open setting is a feature common to all of the research described in this subsection, beginning with the Schlosser work, reflecting its origins in the study of continuous-flow stirred-tank reactors. 

\subsection{A look ahead: stoichiometry, concordance, and strong concordance} 

	In part, fully-open reactors become easier to study because, in that setting, stoichiometric constraints become far less of a concern. Reaction networks typically studied in biology are \emph{not} fully open --- they rarely contain a ``degradation reaction" for \emph{every} species --- and, for such networks, stoichiometric balances come into play more forcefully.  Thus, for example, when asking about the possibility of multiple equilibria, it is usually sensible to ask only about the possibility of two distinct equilibria that are \emph{stoichiometrically compatible}; see Section \ref{sec:CRNTPreliminaries}. For this reason, the fully-open results have limited range. (Nevertheless, they have been used to infer the absence of multiple stoichiometrically compatible equilibria in more general settings \cite{craciun_multiple_2006-2,craciun_multiple_2010}, but the most incisive of the resulting theorems \cite{craciun_multiple_2010} have heretofore been restricted to mass-action kinetics.) 	

	Unlike the more narrowly focused fully-open theory described in \S \ref{sec:fullyopen}, the deficiency-oriented theory of \S \ref{sec:deficiencyresults} was fashioned to deal with stoichiometric-compatibility issues from the outset. For example, when the kinetics is mass-action, the Deficiency Zero and Deficiency One Theorems assert the impossibility of two distinct \emph{stoichiometrically compatible} positive equilibria, regardless of rate constant values, for \emph{any} reaction network that satisfies their hypotheses. No degradation reactions are required to be present.

	As we shall see in Section \ref{sec:ConcordantNetworks}, the definition of network \emph{concordance} has a built-in pre-disposition to stoichiometric compatibility; \emph{it is there from the very beginning}. Only in Section \ref{sec:FullyOpenSection}  will  we be concerned with the concordance of a reaction network's ``fully open extension" --- that is, the network formed by adjoining to the original network a full supply of degradation reactions --- but this is because we can sometimes infer from concordance of the fully open extension properties of the original network. (This is similar in spirit to  \cite{craciun_multiple_2010}, but here there is no requirement that the kinetics be mass-action.) 

	In a sense, then, the theory of concordant networks shares with the deficiency-oriented theory (\S \ref{sec:deficiencyresults}) its indifference to the ``fully open requirement," but it also shares with the fully open theory (\S \ref{sec:fullyopen}) an affinity for special inferences that can be made when there is a full supply of degradation reactions.

	We close this section with an explanation of our choice to defer, completely, the notion of \emph{strong concordance} to the article's end:  To the extent that they can be compared, the very broad \emph{NAC} kinetics admitted by Banaji and Craciun is significantly broader than the (also broad) \emph{weakly monotonic} kinetics we study in most of this article. (Both subsume mass action kinetics.) We show among other things  in Section \ref{sec:ConcordantNetworks} that, when the kinetics is  weakly monotonic, concordance of a reaction network suffices for the impossibility of multiple \emph{stoichiometrically compatible} equilibria.  However, in Section \ref{sec:generalizations} we show that the same can  also be asserted for what we call \emph{two-way weakly monotonic kinetics}, which is very similar NAC kinetics, provided that the network is \emph{strongly} concordant.  (In both cases, there is no requirement that the reactor be fully open.)

	We chose to center this article around the weaker and broader notion of concordance because it already gives rise to a variety of interesting and important consequences, some having nothing at all to do with either choice of kinetics. By deferring to the article's end the discussion of strong concordance and its implications for kinetics of the kind studied by Banaji and Craciun, we tried to keep the already powerful consequences of simple concordance clearly at center stage.

\
\section{Reaction network theory preliminaries}
\label{sec:CRNTPreliminaries}

Here we provide the reaction network theory concepts required for stating and proving the results in this article. Much of this section is also contained in a more extended form, with more discussion and examples, in \cite{feinberg_lectureschemical_1979}. We begin with notational conventions.

\subsection{Notation}
In the study of networks, it is sometimes useful to consider vectors whose components are real numbers associated with the nodes or with the edges of a directed graph. Thus, one typically has a set $I$ (of edges, for example) and the components  are numbers of the form $\{x_i : i \in I \}$. To vectorize such an assignment by working in $\mathbb{R}^N$, where $N$ is the number of elements in $I$, requires that one impart an artificial and sometimes awkward cardinal order to $I$'s elements. Instead, it is useful to simply associate $\{x_i : i \in I \}$ in the obvious way with a vector $x$ in the vector space of real-valued functions having domain $I$. 

When $I$ is a finite set, we denote the vector space of real-valued functions with domain $I$ by $\mathbb{R}^I$. The subset of $\mathbb{R}^I$ consisting of those functions that take only positive (nonnegative) values is denoted $\mathbb{R}_+^I\: (\overline{\mathbb{R}}_+^I)$. For $x \in \mathbb{R}^I$ and $i \in I$,  the symbol $x_i$ denotes the value assigned to $i$ by $x$. 

For each $x \in \mathbb{R}^I$, the symbol $\exp(x)$ denotes the element of $\mathbb{R}_+^I$ defined by \[(\exp(x))_i := \exp x_i, \hspace{0.2 cm} \forall i \in I. \] For each $x \in \mathbb{R}^I$ and for each $z \in \overline{\mathbb{R}}_+^I$, the symbol $x^z$ denotes the real number defined by \[x^z := \prod_{i \in I} (x_{i})^{z_i}, \] with the understanding that $0^0 = 1.$ For each $x, x' \in \mathbb{R}^I$, the symbol $x \circ x'$ denotes the element of $\mathbb{R}^I$ defined by \[ (x \circ x')_i := x_i x'_i, \hspace{0.2 cm} \forall i \in I.\] For each $z  \in \mathbb{R}_+^I$, the symbol $\frac{1}{z}$ denotes the element of $\mathbb{R}_+^I$  defined by \[\left(\frac{1}{z}\right)_i := \frac{1}{z_i}, \hspace{0.2 cm} \forall i \in I. \]
\noindent	
For each $i \in I$, we denote by $\omega_i$ the element of $\mathbb{R}^I$ such that $(\omega_i)_j = 1$ whenever $j=i$ and $(\omega_i)_j = 0$ whenever $j \neq i$.

The \emph{standard basis} for $\mathbb{R}^I$ is the set $\left\{\omega_i \in \mathbb{R}^I : i \in I\right\}$. Thus, for each $x \in \mathbb{R}^I$, we have the representation $x = \sum_{i \in I} x_i \omega_i$. The \emph{standard scalar product} in $\mathbb{R}^I$ is defined as follows: If $x$ and $x'$ are  elements  of $\mathbb{R}^I$, then \[ x \cdot x' = \sum_{i \in I} x_i x'_i .\] Note that the standard basis of $\mathbb{R}^I$ is orthonormal with respect to the standard scalar product. It will be understood that $\mathbb{R}^I$ carries the standard scalar product and the norm derived from the standard scalar product. It will also be understood that $\mathbb{R}^I$ carries the corresponding norm topology. 

Whenever $U$ is a linear subspace of $\mathbb{R}^I$, we denote by $U^\perp$ the orthogonal complement of $U$ in $\mathbb{R}^I$ with respect to the standard scalar product.

By the \emph{support} of $x \in \mathbb{R}^I$, denoted $\supp x$, we mean the set of indices $i \in I$ for which $x_i$ is different from zero. When $\xi$ is a real number, the symbol $\sgn(\xi)$ denotes the sign of $\xi$. When $x$ is an element of $\mathbb{R}^I$, $\sgn(x)$ denotes the function with domain $I$ defined by \[(\sgn(x))_i := \sgn x_i, \hspace{0.2 cm} \forall i \in I. \]

\subsection{Some definitions}

	We will use the reaction network displayed in \eqref{EQ:ExampleNetwork} to motivate some of our definitions. Following Horn and Jackson \cite{horn_general_1972}, we call the objects at the heads and tails of the reaction arrows --- $2A, B, C, C+D$ and $E$ in \eqref{EQ:ExampleNetwork} --- the \emph{complexes} of the network. In this way, we can view the network as a directed graph, with complexes playing the role of the vertices and reaction arrows playing the role of the edges.
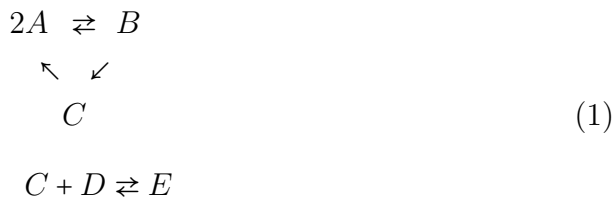
\begin{align}
\label{EQ:ExampleNetwork}
\nonumber &\hspace{-2mm}2A \hspace{2 mm} \rightleftarrows \hspace{1.5 mm} B \\
\nonumber & \hspace{.5 mm} \nwarrow \hspace {3 mm}  \swarrow \\
& \hspace{5 mm} C \\[0.75 em]
\nonumber & C + D \rightleftarrows E
\end{align}

\begin{rem}
We shall frequently be working in \RS, where \scrS \; is the set of species in a network. In this special case, it is advantageous to replace symbols for the standard basis of \RS  \  with the names of the species themselves. Thus, if the species in the network are given by $\scrS = \{A,B,C,D,E\}$ then a vector such as $\omega_C +\omega_D \in \RS$ can instead be written as $C+D$, and $2\omega_A$ can be written as $2A$. In fact, \RS\, can then be identified with the vector space of formal linear combinations of the species.  In this way, the complexes of a reaction network with species set \scrS \;can be identified with vectors in \RS.

\end{rem}

\begin{definition}
\label{DEF:ChemicalReactionNetwork}
A \emph{chemical reaction network} consists of three finite sets:
\begin{enumerate}
	\item a set $\mathscr{S}$ of distinct \emph{species} of the network;
	\item a set $\mathscr{C} \subset \bar{\mathbb{R}}_+^{\mathscr{S}}$ of distinct \emph{complexes} of the network;
	\item a set $\mathscr{R} \subset \mathscr{C} \times \mathscr{C}$ of distinct \emph{reactions}, with the following properties:
		\begin{enumerate}
			\item $(y,y) \notin \mathscr{R}$ for any $y\in \mathscr{C}$;
			\item for each $y \in \mathscr{C}$ there exists $y' \in \mathscr{C}$ such that $(y,y') \in \mathscr{R}$ or such that $(y',y) \in \mathscr{R}$.
		\end{enumerate}
	\end{enumerate}
\end{definition}
\smallskip

If $(y,y')$ is a member of the reaction set $\mathscr{R}$, we say that $y$ \emph{reacts to} $y'$, and, following the usual notation in chemistry, we write $y \rightarrow y'$ to indicate the reaction whereby complex $y$ reacts to complex $y'$. We call the complex situated at the tail of a reaction arrow the \emph{reactant complex} of the corresponding reaction, and the complex situated at the head of a reaction arrow the reaction's a \emph{product complex}.
\smallskip

	The set of species of the network depicted in \eqref{EQ:ExampleNetwork} is $\scrS = \{A, B, C, D, E\}$. The set of complexes of the network is $\scrC = \{2A, B, C, C + D, E\}$. The set of reactions of the network is $\scrR = \{2A \to B, B \to C, C \to 2A, C + D \to E, E \to C + D\}$.
\smallskip

	The diagram in \eqref{EQ:ExampleNetwork} is an example of a \emph{standard reaction diagram}: each complex in the network is displayed precisely once, and each reaction in the network is indicated by an arrow in the obvious way. In very few places in this article we will refer to the \emph{linkage classes} of a reaction network, which for our purposes can be identified with the connected components of the network's standard reaction diagram. Thus, in network \eqref{EQ:ExampleNetwork} there are two linkage classes, containing, respectively, the complexes $\{2A, B, C\}$ and $\{C + D, E\}$. For a formal definition of a linkage class see \cite{feinberg_lectureschemical_1979}.

	More important is the idea of \emph{weak reversibility}. The following definition provides some preparation.    

\begin{definition}
\label{DEF:UltimatelyReactsTo}
A complex $y \in \mathscr{C}$ \emph{ultimately reacts} to a complex $y' \in \mathscr{C}$ if any of the following conditions is satisfied:
\begin{enumerate}
\item{$y \rightarrow y' \in \mathscr{R}$};
\item{There is a sequence of complexes $y(1), y(2), \ldots , y(k)$ such that \[ y \rightarrow y(1) \rightarrow y(2) \rightarrow \ldots \rightarrow y(k) \rightarrow y'.\]}
\end{enumerate}
\end{definition}
\smallskip

In our example, the complex $2A$ ultimately reacts to the complex $C$, but the complex $C$ does not ultimately react to the complex $C + D$. 

\begin{definition}
\label{EQ:WeaklyReversible}
A reaction network \rnet is called \emph{weakly reversible} if for each $y, y' \in \mathscr{C}$, $y' \mbox{ ultimately reacts to } y$ whenever $y \mbox{ ultimately reacts to } y'$. A network is called \emph{reversible} if $y' \rightarrow y \in \mathscr{R}$ whenever $y \rightarrow y' \in \mathscr{R}$.
\end{definition}

Network \eqref{EQ:ExampleNetwork} is an example of a weakly reversible reaction network that is not reversible. Note that any  reversible  network is also weakly reversible. Note also that whenever a weakly reversible reaction network is displayed as a standard reaction diagram, every arrow in the diagram belongs to  a directed cycle of arrows.
\smallskip

\begin{definition}
\label{DEF:Reaction vectors}
The \emph{reaction vectors} for a reaction network \rnet\  are the members of the set \[ \left\{y' - y \in \RS : y \rightarrow y' \in \mathscr{R} \right\}. \]
\end{definition}
\smallskip

In our example the reaction vector corresponding to the reaction $2A \to B$ is $B - 2A$, the reaction vector corresponding to the reaction $C + D \to E$ is $E - C - D$, and so on.

\begin{definition}
\label{DEF:StoichiometricSubspace}
The \emph{stoichiometric subspace} $S$ of a reaction network \rnet \  is the linear subspace of $\RS$ defined by 

\begin{equation}
S := \Span \left\{ y' - y \in \RS : y \rightarrow y' \in \mathscr{R} \right\}.
\label{EQ:StoichiometricSubspace}
\end{equation}
\end{definition}
\smallskip

We note that, for a reaction network \rnet, the stoichiometric subspace $S$ will often be a proper subspace of $\RS$. In other words, it will often be the case that the dimension of $S$ will be smaller than the number of species in the network. For example, in network \eqref{EQ:ExampleNetwork} we have $\dim S = 3$ while $\dim \RS = \#(\scrS) = 5$. In fact, the stoichiometric subspace will be a proper subspace of \RS \, whenever the network is \emph{conservative}:

\begin{definition}
\label{DEF:Conservative}
A reaction network \rnet \ is \emph{conservative} whenever the orthogonal complement $S^\perp$ of the stoichiometric subspace $S$ contains a strictly positive member of $\RS$: \[S^\perp \cap \PS \neq \emptyset.\]
\end{definition}

Network \eqref{EQ:ExampleNetwork}, for example, is conservative: it is easy to verify that the strictly positive vector $(A + 2B + 2C + D + 3E) \in \PS$ is orthogonal to each of the reaction vectors of \eqref{EQ:ExampleNetwork}.
\bigskip

If \rnet\  is a reaction network, then a mixture state will generally be represented by a \emph{composition} $c \in \PbarS$, where, for each $s \in \mathscr{S}$, we understand $c_s$ to be the molar concentration of $s$. By a \emph{positive composition} we mean a strictly positive composition --- that is, a composition in $\PS$.

\begin{definition}
\label{DEF:Kinetics}
A \emph{kinetics} $\mathscr{K}$ for a reaction network \rnet\ is an assignment to each reaction $y \to y' \in \scrR$ of a \emph{rate function} $\mathscr{K}_{y \to y'}:\PbarS \to \overline{\mathbb{R}}_+$ such that \[ \mathscr{K}_{y \to y'}(c) > 0 \mbox{ if and only if } \supp y \subset \supp c.\]
\end{definition}
\smallskip

\begin{definition}
\label{DEF:KineticSystem}
A \emph{kinetic system} \kinsys\  is a reaction network \rnet\  taken with a kinetics $\mathscr{K}$ for the network.
\end{definition}
\smallskip

\begin{rem} In chemical reaction network theory the rate functions are generally presumed at the outset to be continuous, but we have not insisted on that here. This is because some of the main results in this article (in particular those having to do with injectivity) do not require it. For certain other formally stated results, we will sometimes say specifically that we are assuming continuity or differentiability. On the other hand, when we discuss in an informal way the differential equations associated with a kinetic system, it will be understood implicitly that there is smoothness sufficient to ensure the that the commonly presumed properties of differential equations are present. 
\end{rem}
\medskip
In this article we shall often refer to \emph{mass action kinetics} and to \emph{mass action kinetic systems}. Both are formally defined below:

\begin{definition}
\label{DEF:MassActionKinetics}
A kinetics $\scrK$ for a reaction network \rnet is \emph{mass action} if, for each reaction $y \to y' \in \scrR$, there is a positive number $k_{y \to y'}$ such that the rate function $\scrK_{\rxn}$ takes the form
\begin{equation}
\label{EQ:MassActionKinetics}
\scrK_{y \to y'}(c) = k_{y \to y'} c^y.
\end{equation}

\noindent The positive number $k_{y \to y'}$ is the \emph{rate constant} for reaction $y \to y'$. 
\end{definition}
\smallskip

\begin{definition}
\label{DEF:MassActionKineticSystem}
A \emph{mass action kinetic system} is a reaction network taken together with a mass action kinetics for the network.
\end{definition}
\smallskip

\begin{definition}
\label{DEF:SpeciesFormationRateFunction}
The \emph{species formation rate function} for a kinetic system \kinsys\  with stoichiometric subspace $S$ is the map $f:\PbarS \to S$ defined by

\begin{equation}
\label{EQ:SpeciesFormationRateFunction}
f(c) = \sum_{\rxn \in \scrR} \scrK_{y \to y'} (c) (y' - y). 
\end{equation}
\end{definition}
\smallskip
For a kinetic system \kinsys\  whose underlying network is our example network \eqref{EQ:ExampleNetwork}, the species formation rate function has the following species-wise form:

\begin{align}
\label{EQ:CoordinateForm}
&f_A(c) = -2\scrK_{2A \to B}(c) +2\scrK_{B \to 2A}(c) + 2\scrK_{C \to 2A}(c), \\
\nonumber &f_B(c) = -\scrK_{B \to C}(c) + \scrK_{2A \to B}(c) - \scrK_{B \to 2A}(c) , \\
\nonumber &f_C(c) = -\scrK_{C \to 2A}(c) +\scrK_{B \to C}(c) -\scrK_{C + D \to E}(c) + \scrK_{E \to C + D}(c), \\
\nonumber &f_D(c) = -\scrK_{C + D \to E}(c) + \scrK_{E \to C + D}(c), \\
\nonumber &f_E(c) = -\scrK_{E \to C + D}(c) + \scrK_{C + D \to E}(c).
\end{align}

\begin{definition}
\label{DEF:DifferentialEquation}
The \emph{differential equation} for a kinetic system with species formation rate function $f(\cdot)$ is given by

\begin{equation}
\label{EQ:DifferentialEquation}
\dot{c} = f(c).
\end{equation} 
\end{definition}
\smallskip

Let \kinsys\  be a kinetic system. From equations \eqref{EQ:StoichiometricSubspace}, \eqref{EQ:SpeciesFormationRateFunction}, and \eqref{EQ:DifferentialEquation} we observe that the vector $\dot{c}$ will invariably lie in the stoichiometric subspace $S$ of the network \rnet. Thus, the difference of any two compositions $c \in \PbarS$ and $c' \in \PbarS$ that lie  along   the same solution of (\ref{EQ:DifferentialEquation}) will always reside in $S$. This motivates the following definition: 

\begin{definition}
\label{DEF:StoichiometricallyCompatible}
Let \rnet\  be a reaction network with stoichiometric subspace $S$. Two compositions $c$ and $c'$ in $\PbarS$ are called \emph{stoichiometrically compatible} if $c' - c \in S$.
\end{definition}
\smallskip

We note that stoichiometric compatibility is an equivalence relation. As such, it partitions $\PbarS$ into equivalence classes that we call  \emph{stoichiometric compatibility classes}.   Thus, the stoichiometric compatibility class containing an arbitrary composition $c$, denoted $(c + S) \cap \PbarS$, is given by

\begin{equation}
\label{EQ:StoichiometricCompatibilityClass}
(c + S) \cap \PbarS = \left\{ c' \in \PbarS: c' - c \in S \right\}.
\end{equation}
We observe, as the notation suggests, that $(c + S) \cap \PbarS$ is the intersection of $\PbarS$ with the parallel of $S$ containing $c$. 

	A stoichiometric compatibility class will typically contain a wealth of (strictly) positive compositions. We say that a stoichiometric compatibility class is \emph{nontrivial} if it contains a member of \PS. To see that a stoichiometric compatibility class can be trivial, consider the simple reaction network \mbox{$A+B \rightleftarrows C$}, and let $\bar{c}$ be the composition defined by $\bar{c}_A = 1,  \bar{c}_B = 0, \bar{c}_C = 0$. Then the stoichiometric compatibility class containing $\bar{c}$ has $\bar{c}$ as its only member.

\begin{definition}
\label{DEF:Equilibrium}
\noindent An \emph{equilibrium} of a kinetic system \kinsys\  is a composition $c \in \PbarS$ for which $f(c) = 0$. A \emph{positive equilibrium} of a kinetic system \kinsys\  is an equilibrium that lies in $\PS$.
\end{definition}
\smallskip

	In light of Definition \ref{DEF:Kinetics}, a kinetic system can admit a positive equilibrium only if its reaction vectors are positively dependent:

\begin{definition}
\label{DEF:PositivelyDependent} The reaction vectors for a reaction network \rnet\  are \emph{positively dependent} if for each reaction $y \to y' \in \scrR$ there exists a positive number $\kappa_{y \to y'}$ such that
\begin{equation}
\sum_{\rxn \in \scrR}\kappa_{\rxn}(y' - y) = 0.
\end{equation}
\end{definition}

\begin{rem}
For any weakly reversible network, the reaction vectors are positively dependent \cite{feinberg_lectureschemical_1979}.
\end{rem}
\medskip

Occasionally we will want to consider changes in the value of the species formation rate function corresponding to small departures from a positive composition $c^*$ toward nearby compositions that are stoichiometrically compatible with $c^*$. Thus, for a kinetic system \kinsys\ with stoichiometric subspace $S$, with smooth reaction rate functions, and with species formation rate function  $f:\PbarS \to S$, we will want to work with  the derivative  $df(c^*): S \to S$, given by

\begin{equation}
\label{EQ:Derivative}
df(c^*)\sigma = \left. \frac{df(c^* + \theta \sigma)}{d\theta} \right|_{\theta = 0}, \quad  \forall \sigma \in S.
\end{equation}

\medskip \noindent
In this case, we say that $c^* \in \PS$ is a \emph{degenerate equilibrium} if $c^*$ is an equilibrium and if, moreover, $df(c^*)$ is singular.

\section{Concordant networks, weakly monotonic\\ kinetics, and injective kinetic systems}

\label{sec:ConcordantNetworks}

In preparation for the definition of concordance we consider a reaction network \rnet\  with stoichiometric subspace $S \subset \RS$, and we let $L: \RR \to S$ be the linear map defined by
\begin{equation}
\label{eq:LDefinition}
L\alpha = \sum_{\rxn \in \scrR}\alpha_{\rxn}(y' - y).
\end{equation}

\begin{definition} The reaction network \rnet\  is \emph{concordant} if there do not exist an $\alpha \in \ker L$ and a nonzero $\sigma \in S$ having the following properties:\\

\noindent (i) For each $\rxn \in \scrR$\; such that $\alpha_{\rxn} \neq 0$, \supp $y$ contains a species $s$ for which $\mathrm{sgn}\, \sigma_s = \mathrm{sgn}\, \alpha_{\rxn}$.\\

\noindent (ii) For each $\rxn \in \scrR$\; such that $\alpha_{\rxn} = 0$, either $\sigma_s = 0$\, for all $s \in \supp y$ or else \supp $y$ contains species $s\, \textrm{and } s'$ for which  $\mathrm{sgn}\, \sigma_s = -\, \mathrm{sgn}\; \sigma_{s'}$, both not zero.
\medskip\noindent\\
A network that is not concordant is \emph{discordant}.
\label{def:Concordant}
\end{definition}

\begin{rem} Concordance is a network property that is discernible by sign-checking in $\ker L$ and $S$.  A Windows-based computer program that implements such a concordance test is freely available \cite{Ji_toolboxWinV21}; it is fast for networks of moderate size.
\label{rem:ConcordanceComputerProgram}
\end{rem}
\smallskip
\noindent 
\emph{Examples.} While network \eqref{EQ:ExampleNetwork} is a fairly simple ``toy'' that is concordant, the network depicted in \eqref{eq:CalciumNet} serves as a more intricate and more biologically-motivated concordant example. In Reidl et al. \cite{reidl_model_2006},  network \eqref{eq:CalciumNet} served to model calcium dynamics in olfactory cilia.  The concordance of \eqref{eq:CalciumNet} was ascertained by means of the sign-checking procedure implemented in \cite{Ji_toolboxWinV21}.
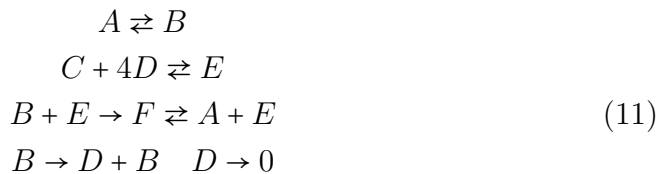
\begin{eqnarray}
\label{eq:CalciumNet}
&A \rightleftarrows B \nonumber\\
&C + 4D \rightleftarrows E \nonumber\\
&B + E \to F \rightleftarrows A + E\\
&B \to D + B \quad D \to 0 \nonumber 
\end{eqnarray}

A biologically motivated example of a discordant network is shown in  \eqref{eq:EnvZ}. This network underlies a model of the bacterial two-component signaling system EnvZ-OmpR \cite{shinar_input_2007,shinar_structural_2010,shinar_dreaded_2011}, which regulates the number and size of membrane pores in response to changes in osmolarity.
\medskip
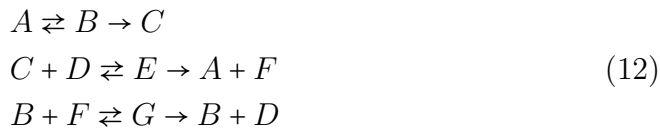
\begin{align}
\label{eq:EnvZ}
\nonumber&A \rightleftarrows B \rightarrow C \\
 &C + D \rightleftarrows E \rightarrow A + F \\
\nonumber &B + F \rightleftarrows G \rightarrow B + D
\end{align}

\begin{definition} A kinetic system \kinsys \ is \emph{injective} if, for each pair of distinct stoichiometrically compatible compositions $c^* \in \PbarS$ and $c^{**} \in \PbarS$, at least one of which is positive,

\begin{equation}
\sum_{\rxn \in \scrR}\scrK_{\rxn}(c^{**})(y' - y) \neq \sum_{\rxn \in \scrR}\scrK_{\rxn}(c^*)(y' - y).
\end{equation}
\end{definition}
\smallskip

\begin{rem}
Clearly, an injective kinetic system cannot admit two distinct stoichiometrically compatible equilibria, at least one of which is positive.
\end{rem}

\begin{definition}
\label{def:WeaklyMonotonic}
A kinetics \scrK\   for reaction network \rnet\  is \emph{weakly monotonic} if, for each pair of compositions $c^*$ and $c^{**}$, the following implications hold for each reaction $\rxn \in \scrR$ such that  $\supp y \subset \supp c^*$ and $\supp y \subset \supp c^{**}$:\\

\noindent (i) $\scrK_{\rxn}(c^{**})\; >\; \scrK_{\rxn}(c^{*})\quad \Rightarrow \quad$ there is a species $s \in \supp y$ with $c_s^{**} > c_s^*$.\\

\noindent (ii) $\scrK_{\rxn}(c^{**})\;=\; \scrK_{\rxn}(c^{*})\quad \Rightarrow \quad$ $c_s^{**} = c_s^*$ for all $s \in \supp y$ or else there are species $s, s' \in \supp y$ with $c_s^{**} > c_s^*$ and $c_{s'}^{**} < c_{s'}^*$.\\

\noindent We say that the kinetic system \kinsys\ is weakly monotonic when its kinetics \scrK\, is weakly monotonic.
\end{definition}

\begin{rem}
\label{rem:weakmonotonicboundary}
Note that if a kinetic system is weakly monotonic and if $c^{**}$ is a positive composition, then the implications (i) and (ii) will hold for \emph{all} reactions in the network, even when $c^*$ is not strictly positive.
\end{rem}

\begin{rem}
Clearly, every mass action kinetic system is weakly monotonic. Note that the definition of weak monotonicity makes no assumptions about the continuity or the smoothness of the kinetic rate functions. When, for a weakly monotonic kinetics, a reaction-rate function $\scrK_{\rxn}(\cdot)$ is differentiable, there is no requirement that its derivative with respect to $c_{s},\;s \in \supp y,$ be strictly positive at all positive compositions. 
\end{rem}

\begin{proposition}
\label{prop:ConcordanceInjectivity}
A weakly monotonic kinetic system \kinsys\  is injective whenever its underlying reaction network \rnet is concordant. In particular, if the underlying reaction network is concordant, then the kinetic system cannot admit two distinct stoichiometrically compatible equilibria, at least one of which is positive.
\end{proposition}

\begin{proof} Suppose on the contrary that, for the weakly monotonic kinetic systems \kinsys, there are distinct stoichiometrically compatible compositions $c^*$ and $c^{**}$, with $c^{**} \in \PS$, such that

\begin{equation}
\sum_{\rxn \in \scrR}\scrK_{\rxn}(c^{**})(y' - y) = \sum_{\rxn \in \scrR}\scrK_{\rxn}(c^*)(y' - y).
\end{equation}
Clearly, then, $\alpha \in \RR$ defined by
\begin{equation}
\alpha_{\rxn} := \scrK_{\rxn}(c^{**}) -  \scrK_{\rxn}(c^{*}), \; \forall \rxn \in \scrR
\end{equation}
is a member of $\ker L$. Because  $c^{*}$ and $c^{**}$ are distinct and stoichiometrically compatible, the vector $\sigma :=  c^{**} - c^{*}$ is a nonzero member of $S$.  Note that if, for a particular reaction \rxn, we have $\alpha_{\rxn}  > 0$, then weak monotonicity and Remark \ref{rem:weakmonotonicboundary} require that, for some species $s \in \supp y$,  $\sigma_s := c_s^{**} - c_s^* >  0$.  A similar argument  can be mounted when we have $\alpha_{\rxn}  < 0$; in this case, there is $s \in \supp y$ such that  $\sigma_s := c_s^{**} - c_s^* <  0$. Finally, if we have $\alpha_{\rxn} = 0$, then weak monotonicity requires that the numbers $\{\sigma_s\}_{s\, \in\, \supp y}$ are either all zero or else there are two of opposite nonzero sign.  The existence of $\alpha$ and $\sigma$, so constructed, contradicts the supposition that the network \rnet\  is concordant.
\end{proof}

\noindent
\emph{Example}. Proposition \ref{prop:ConcordanceInjectivity} tell us that no weakly monotonic kinetics for the complex but concordant network \eqref{eq:CalciumNet} can give rise to multiple stoichiometrically compatible equilibria, at least one of which is positive. This is consistent with the findings in \cite{reidl_model_2006}.
\bigskip

	The simple Proposition \ref{prop:ConcordanceInjectivity} --- one of the most important of this article --- tells us that, for any weakly monotonic kinetic system, injectivity (and therefore the absence of distinct stoichiometrically compatible equilibria, at least one of which is positive) can be precluded merely by establishing concordance of the underlying reaction network, perhaps with the assistance of a concordance-checking computer program such as the one made available in \cite{Ji_toolboxWinV21}. As we shall see in a subsequent article, the Species-Reaction Graph can sometimes also serve to establish network concordance. (See Remark \ref{rem:SRGraphConditions}.)

	The converse of Proposition \ref{prop:ConcordanceInjectivity} is not true. That is, there can be a weakly monotonic kinetic system --- in fact, a mass action system --- that is injective even when its underlying reaction network is not concordant. On the other hand, the following proposition resembles a converse; proof by construction is provided in Appendix A.

\begin{proposition}
\label{prop:partialconverse}
Let \rnet be a discordant reaction network. Then there exists for it a weakly monotonic kinetics \scrK \ such that the resulting kinetic system \kinsys\ is not injective.
\end{proposition}

\begin{rem}[Example] Proposition \ref{prop:partialconverse} asserts the existence of a weakly monotonic kinetics for the discordant network \eqref{eq:EnvZ} that gives rise to a non-injective kinetic system. In fact, it can be shown by direct computation that any mass-action system based on \eqref{eq:EnvZ} possesses at least two stoichiometically compatible equilibria \cite{shinar_input_2007} --- one positive and one non-positive --- thereby violating injectivity.
\medskip

\end{rem}

	Taken together, Propositions \ref{prop:ConcordanceInjectivity} and \ref{prop:partialconverse} can be summarized in the following theorem --- a theorem that tells us that the class of concordant reaction networks is \emph{precisely} the class of networks that are injective for \emph{every} assignment of a weakly monotonic kinetics:

\begin{theorem} 
\label{thm:ConcordanceInjectivityThm}
A reaction network has injectivity in all weakly monotonic kinetic systems derived from it if and only if the network is concordant.
\end{theorem}

	We discuss more subtle connections between concordance and injectivity in the next section.

\medskip
\section{A remark about connections to earlier work} 
\label{sec:Connections}

In \cite{craciun_multiple_2005-1,craciun_multiple_2006-1} and \cite{craciun_understanding_2006-1} attention was focused on \emph{mass action} kinetic systems in which the underlying reaction network contains a ``degradation reaction" of the form $s \to 0$ for each species $s$ in the network. (Here such networks are called \emph{fully open}.) In this context it was shown that there are certain fully open networks having the property that all \emph{mass action} kinetic systems deriving from them are injective --- that is,  regardless of rate constant values. These were called \emph{injective (mass action) networks} in \cite{craciun_multiple_2005-1,craciun_multiple_2006-1}, and it was shown that a network is injective in this (mass action) sense if its \emph{Species-Reaction Graph} satisfies the mild conditions stated in Remark \ref{rem:SRGraphConditions}. 

	With this in mind, it is not unreasonable to suppose that a variant of Proposition \ref{prop:partialconverse} might be true: A network can have the property that every \emph{mass action} system deriving from it is  injective, \emph{regardless of rate constants}, only if the network is concordant. 

	In fact, this is  false. We show in \eqref{eq:counterexample} an example of a fully open reaction network that is not concordant but for which any \emph{mass action} system derived from it is injective. Injectivity follows from a theorem in \cite{craciun_PhDThesis}\footnote{A  similar theorem appeared earlier in \cite{ Schlosser_PhDThesis,schlosser_theory_1994}. In the case of \eqref{eq:counterexample}, injectivity follows from the fact that the Species-Reaction Graph has just one cycle, and it is a \emph{c-cycle} --- that is, a cycle whose edge set is the union of c-pairs.} that invokes Species-Reaction Graph conditions different from those in Remark \ref{rem:SRGraphConditions}.  On the other hand, it can be established that the network is \emph{not} concordant by a choice of $\alpha$ and $\sigma$ that satisfy the conditions in Definition \ref{def:Concordant}.
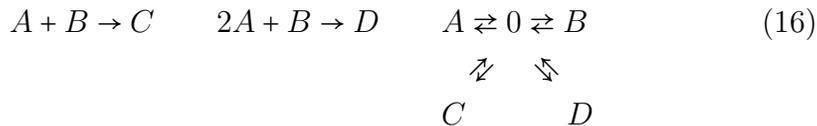
\medskip
\begin{align}
\label{eq:counterexample}
A + B \to C \qquad 2A + B \to D \qquad &A \rightleftarrows 0 \rightleftarrows B \\
& \;\;\neswarrows\quad \senwarrows \nonumber\\
&C \quad \quad \quad \;D \nonumber
\end{align}

\medskip
	The example tells us that ``network injectivity" in the \emph{mass action} case should not be conflated with network concordance. In particular, the example tells us that there are fully open networks for which  \emph{particular} Species-Reaction Graph conditions, \emph{different from those in Remark \ref{rem:SRGraphConditions}}, might ensure ``\emph{mass action} network injectivity" but for which network concordance need not obtain. 

	In a separate article, however, we will show that when a fully open network's Species-Reaction Graph satisfies the specific conditions stated in Remark \ref{rem:SRGraphConditions} --- conditions that ensured injectivity in the mass action setting considered in \cite{craciun_multiple_2005-1,craciun_multiple_2006-1,craciun_understanding_2006-1} and in the far more general NAC kinetics setting considered by Banaji and Craciun \cite{banaji_graph-theoretic_2009,banaji_graph-theoretic_2010} --- \emph{then those same Species-Reaction Graph conditions will ensure that the network itself, divorced from any kinetics at all, is concordant and, in fact, strongly concordant (Section \ref{sec:generalizations})}.  Thus, when the kinetics is weakly monotonic or two-way weakly monotonic (Section \ref{sec:generalizations}), injectivity and the absence of multiple equilibria are ensured.

	In the fully open network setting, then, the Species-Reaction Graph conditions stated in Remark \ref{rem:SRGraphConditions} will imply injectivity for large classes of kinetics through the very simple Propositions \ref{prop:ConcordanceInjectivity} and \ref{prop:StrongConcordanceInjectivity}.  \emph{While earlier work based on these conditions appealed to Jacobians, special determinant attributes, and the Gale-Nikaido theorem to infer injectivity, the elementary Propositions \ref{prop:ConcordanceInjectivity} and \ref{prop:StrongConcordanceInjectivity} make those appeals unnecessary.}   Moreover, we shall see in Section \ref{sec:FullyOpenSection} how properties of the Species-Reaction Graph will often imply concordance of a network (and, therefore, injectivity of any weakly monotonic kinetic system derived from it) even when the network is not fully open.

\section{Persistence in weakly reversible concordant kinetic systems:  Boundary behavior and the existence of positive equilibria}

\label{sec:Persistence}

  It is our aim in this section to establish two related properties inherent to  kinetic systems,  \emph{not necessarily weakly monotonic},  that derive from weakly reversible concordant reaction networks. 

	These properties are connected to ``persistence" questions: Viewed from an ecological perspective, one might be interested to know for a kinetic system whether there will be a positive equilibrium --- that is, an equilibrium in which all species coexist. And one might also want to know whether a population that begins at a state in which all species are present can evolve to a different one --- in particular to an equilibrium --- in which one or more species are extinct (i.e., absent). For kinetic systems in which the underlying network is weakly reversible, questions like these become especially challenging.  

	As we shall see, however, a weakly reversible \emph{concordant} reaction network invariably has pleasant properties that make these questions far easier to resolve:

	First, we examine behavior of the species formation rate function on the boundary of a nontrivial stoichiometric compatibility class. (Recall that by a \emph{nontrivial stoichiometric compatibility class} we mean one that contains a strictly positive composition.) In particular, we show that, on the boundary, \emph{there can be no equilibria at all}.  In fact, at every boundary composition $\bar{c}$, the species formation rate vector $f(\bar{c})$ points into the stoichiometric compatibility class in the following sense: There is at least one species $s^*$ for which $\bar{c}_{s^*} = 0$ while $\dot{c}_{s^*} = f_{s^*}(\bar{c}) > 0$.  (Thus, the inward-pointing boundary vector field thwarts an approach to the boundary along a trajectory that originates at a strictly positive composition.) 

	Second, we show that if the network is \emph{conservative} and the kinetic rate functions are continuous then each nontrivial stoichiometric compatibility class contains \emph{at least one} positive equilibrium. If, in addition, the kinetics is weakly monotonic, then there is \emph{precisely one} equilibrium in each nontrivial stoichiometric compatibility class, and it is positive. (Of course, in the weakly monotonic case, uniqueness already follows from Proposition \ref{prop:ConcordanceInjectivity}.)
\bigskip

\subsection{Preliminaries: reaction transitive compositions} We begin with some old reaction network theory ideas, derived from \cite{feinberg_chemical_1987}\footnote{See in particular Propositions 5.3.1-5.3.2, Remark 6.1.E, and Appendices I and II in \cite{feinberg_chemical_1987}. The idea behind Def. \ref{def:reaction_transitive} was important not only in \cite{feinberg_chemical_1987} but also in more recent papers, where the formulation is phrased differently. A species set is called  \emph{semi-locking}  in \cite{anderson_global_2008} or a \emph{siphon} in \cite{angeli_persistence_2011} if its complement in the full set of species is  reaction-transitive.}, that are independent of concordance. The following definition is intended to convey a simple idea: A species set $\scrS^* \subset \scrS$ is \emph{reaction-transitive} if, whenever all the \emph{reactant} species of a reaction are present in $\scrS^*$, then so too are that reaction's \emph{product} species.

\begin{definition}
\label{def:reaction_transitive}
 For a reaction network \rnet, a species set $\scrS^* \subset \scrS$ is \emph{reaction-transitive} if, for all $\rxn \in \scrR$,
\begin{equation}
\mathrm{supp}\,y \subset \scrS^* \; \Rightarrow \; \mathrm{supp}\,y' \subset  \scrS^*.   
\end{equation} 
We say that a composition $c \in \PbarS$ is reaction-transitive if $\supp c$ is reaction transitive.
\end{definition}

	Clearly every positive composition is reaction-transitive. It follows from properties presumed of a kinetics in Section \ref{sec:CRNTPreliminaries} that if $\bar{c}$ is a composition on the boundary of $\PbarS$ and if species $s$ is ``missing" at composition $\bar{c}$ (i.e., if $\bar{c}_s = 0$), then the
 formation rate of $s$ at composition $\bar{c}$ is not negative  (i.e., $f_s(\bar{c}) \ge 0)$. The following lemma tells us that if $\bar{c}$ is not 
reaction-transitive, then there must be a missing species which, at composition $\bar{c}$, is produced at a \emph{strictly positive} rate. (See \cite{feinberg_lectureschemical_1979} or Appendix I in \cite{feinberg_chemical_1987}.)  

\begin{lemma} 
\label{lemma:not_reaction_transitive}
Consider a kinetic system \kinsys\ with species formation rate function $f(\cdot)$. If $\bar{c}$ is a composition on the boundary of  $\PbarS$, then  $f_s(\bar{c}) \ge 0$ for each species $s$ such that $\bar{c}_s = 0$. Moreover, if $\bar{c}$ is not reaction-transitive, then there is a species $s^*$ such that $\bar{c}_{s^*} = 0$ and \hbox{$f_{s^*}(\bar{c}) > 0$}. Thus, every equilibrium of the kinetic system is reaction transitive. 
\end{lemma}

\subsection{Boundary behavior for weakly reversible concordant networks}	With this as background, we are in a position to state  a theorem about weakly reversible concordant reaction networks. \emph{Note that the theorem makes mention only of network properties; there is no mention of kinetics.}

\begin{theorem}
\label{thm:concordant_reaction-transitive}
For a weakly reversible concordant reaction network, no nontrivial stoichiometric compatibility class can have on its boundary a composition that is reaction-transitive.
\end{theorem}

Before proving Theorem \ref{thm:concordant_reaction-transitive}, we observe that the theorem and Lemma \ref{lemma:not_reaction_transitive} give us the following corollary:

\begin{corollary}
\label{corr:concordant_no-boundary-equil}
For a kinetic system that derives from a weakly reversible concordant reaction network, no composition on the boundary of a nontrivial stoichiometric compatibility class is an equilibrium. In fact, at each such boundary composition the species formation rate vector points into the stoichiometric compatibility class in the sense that there is an absent species produced at strictly positive rate.
\end{corollary}

	As might be expected, a species-formation-rate function that, in the sense indicated, points inward on the boundary of a non-trivial stoichiometric compatibility class repels an approach to the boundary by a bounded trajectory that begins at a strictly positive composition. For a formal proof that its $\omega$-limit set is, in fact, disjoint from the boundary when the kinetics is sufficiently smooth see \cite{angeli_persistence_2011}.

\begin{rem} \emph{(About the proof of Theorem  \ref{thm:concordant_reaction-transitive}}) Although Theorem \ref{thm:concordant_reaction-transitive} is about network properties alone, divorced from any mention of kinetics, in its proof we will endow the reversible concordant network of the theorem with a mass action kinetics in order to produce a contradiction. To begin, we will presume that there can be a nontrivial stoichiometric compatibility class that has a reaction-transitive composition $\bar{c}$ on its boundary. Then we will endow the network with a mass action kinetics, with rate constants chosen to make the boundary composition $\bar{c}$ an equilibrium. 

	Now suppose we could take for granted that, for every weakly reversible network and for every assignment of rate constants, the corresponding mass action species-formation rate function admits a \emph{positive} equilibrium in each nontrivial stoichiometric compatibility class. In that case, we would  have a contradiction of Proposition \ref{prop:ConcordanceInjectivity}. In particular, in the nontrivial stoichiometric compatibility class containing $\bar{c}$ we would have, for the mass action system constructed as described, a positive equilibrium that is stoichiometrically compatible with the boundary equilibrium $\bar{c}$. 

	In fact, an unpublished manuscript by Deng et al. \cite{deng_steadystates} does indeed provide an argument for the existence, in each nontrivial stoichiometric compatibility class, of a positive equilibrium for any mass action system in which the underlying network is weakly reversible. In the absence of a published version, however, we can instead mount a similar but somewhat more complicated argument that relies on the same strategy, this time involving construction of mass action system that is \emph{complex balanced} at the boundary composition $\bar{c}$. This would permit invocation of known properties of complex balanced mass action systems, in particular the distribution of their positive equilibria.

	Because the argument that invokes \cite{deng_steadystates} is more straightforward, we give it here in the main text. The alternative argument, which does not rely on \cite{deng_steadystates}, is provided in Appendix B. 
\end{rem}

\begin{rem} 
\label{rem:PrepForBoundaryThm}
In preparation for the nearby proof of Theorem \ref{thm:concordant_reaction-transitive} (and of the proof in the Appendix), it will be helpful if we make the following observation \cite{feinberg_chemical_1987}: Consider a \emph{weakly reversible} network \rnet, and suppose that a composition $c^*$ is reaction-transitive. It is not difficult to see that, if $y^{\dagger}$ and $y^{\dagger \dagger}$ are complexes belonging to the same linkage class, then there are two possibilities: Either $\supp y^{\dagger}$ and $\supp y^{\dagger \dagger}$ are \emph{both} contained in $\supp c^*$ or else \emph{neither} is. 

	This is to say that the linkage classes are of two distinct kinds: Relative to the reaction-transitive composition $c^*$, a \emph{supported linkage classes} is one in which the support of \emph{every} complex is contained in $\supp c^*$, while  an \emph{unsupported linkage class} is one in which the support of \emph{no} complex is contained in $\supp c^*$. 

	If the network is given a kinetics, it follows from properties of a kinetics given in Section \ref{sec:CRNTPreliminaries} that, \emph{at the reaction-transitive composition $c^*$, all reactions associated with a supported linkage class will proceed at strictly positive rate, while all reactions associated with an unsupported linkage class will have zero rate}.  In rougher terms, at composition $c^*$,  all reactions associated with a supported linkage class will be ``turned on" while those associated with an unsupported linkage class will be ``turned off." Moreover, the ``turned on" reactions (if there are any) constitute a \emph{weakly reversible} subnetwork of the original network.
\end{rem} 
\smallskip
\begin{proof}[Proof of Theorem \ref{thm:concordant_reaction-transitive}] Suppose that \rnet \  is a concordant weakly reversible reaction network and that, contrary to what is to be proved, $c^* \in \PbarS$ is a reaction-transitive composition on the boundary of a nontrivial stoichiometric compatibility class. Let
\begin{equation}
\scrC ^* := \{y \in \scrC : \supp y \subset \supp c^*\} \; \mathrm{and} \; 
\scrR ^* := \{ y \to y' \in \scrR : \supp y \subset \supp c^*\}. \nonumber
\end{equation}
Provided that $\scrC ^*$ is not empty,  \{$\scrS, \scrC ^*, \scrR ^*$\} is a weakly reversible subnetwork of the original network. In this case, there are positive numbers $\{\kappa_{\rxn}\}_{\rxn \in \scrR ^*}$ that satisfy \cite{feinberg_lectureschemical_1979}
\begin{equation}
\label{eq:KappaEqn}
\sum_{\rxn \in \scrR ^*}\kappa_{\rxn}(y' - y) = 0.
\end{equation}
Now choose rate constants $\{k_{\rxn}\}_{\rxn \in \scrR ^*}$ for reactions in $\scrR ^*$ to satisfy
\begin{equation}
\label{eq:RateConstChoice}
\kappa_{\rxn} = k_{\rxn}(c^*)^y, \quad \forall \rxn \in \scrR ^*.
\end{equation}
In any case, for the remaining reactions choose rate constants $\{k_{\rxn}\}_{\rxn \in \scrR \setminus \scrR ^*}$ to be any positive numbers. Note that for any $\rxn \in \scrR \setminus \scrR ^*$ we will have $ k_{\rxn}(c^*)^y = 0$ because $\supp y \not\subset \supp c^*$. From this, \eqref{eq:KappaEqn}, and \eqref{eq:RateConstChoice} it follows that
\begin{equation}
\sum_{\rxn \in \scrR ^*}k_{\rxn}(c^*)^y(y' - y) \;\;+ \sum_{\rxn \in \scrR \setminus \scrR ^*}k_{\rxn}(c^*)^y(y' - y) = 0.
\end{equation} 

	This is to say that $c^*$ is a (boundary) equilibrium of the mass action system \mas\, so constructed. Because \rnet is weakly reversible, \cite{deng_steadystates} asserts, for the same mass action system, the existence of a strictly positive equilibrium $c^{**}$  in the nontrivial stoichiometric compatibility class containing $c^*$.  Because the network is concordant, this contradicts the conclusion of Proposition \ref{prop:ConcordanceInjectivity}.
\end{proof}

\subsection{Conservative weakly reversible concordant networks and the existence of positive equilibria} When, in the following theorem, we say that a kinetics is continuous we mean that its various reaction rate functions are all continuous. 

\begin{theorem} 
\label{thm:PositiveEquilibrium}
If \scrK\,is a continuous kinetics for a conservative reaction network \rnet, then the species formation rate function for the kinetic system \kinsys \  has an equilibrium within each stoichiometric compatibility class. If the network is weakly reversible and concordant, then within each nontrivial stoichiometric compatibility class there is a positive equilibrium. If, in addition, the kinetics is weakly monotonic, then that positive equilibrium is the only equilibrium in the stoichiometric compatibility class containing it.
\end{theorem}

	We begin with a lemma that is required for application of a fixed point theorem due to Browder \cite{Browder1967}. For any kinetic system, the lemma merely makes more firmly geometrical the idea that, on the boundary of a stoichiometric compatibility class, the species formation rate vector does not point outward. 

\begin{lemma}
\label{lemma:LemmaForBrowder}
Consider a kinetic system \kinsys\  with stoichiometric subspace $S$ and species formation rate function $f: \PbarS \to S$. If $\bar{c}$ is a composition on the boundary of \PbarS, then there is a composition $c^{\dagger}$ in the stoichiometric compatibility class $(\bar{c} + S) \cap \PbarS$ and a positive number $\lambda$ such that $f(\bar{c}) = \lambda(c^{\dagger} - \bar{c})$.
\end{lemma}

\begin{proof} Choose $\theta > 0$ sufficiently small as to ensure that
\begin{equation}
\label{eq:cdagger1}
\bar{c}_s + \theta f_s(\bar{c}) > 0, \quad \forall s \in \supp \bar{c}.
\end{equation}
For $s \notin \supp \bar{c}$ we have $\bar{c}_s = 0$ and, by Lemma \ref{lemma:not_reaction_transitive}, $f_{s}(\bar{c}) \ge 0$, so 
\begin{equation}
\label{eq:cdagger2}
\bar{c}_s + \theta f_s(\bar{c})  \ge 0, \quad  \forall s \notin \supp\bar{c}.
\end{equation}
Now let $c^{\dagger} := \bar{c} + \theta f(\bar{c})$. From \eqref{eq:cdagger1} and  \eqref{eq:cdagger2} it follows that $c^{\dagger}$ is a member of \PbarS. Because $f(\bar{c})$ is a member of $S$, it follows that  $c^{\dagger}$ is a member of the stoichiometric compatibility class containing $\bar{c}$. Taking $\lambda = \frac{1}{\theta}$, we get the desired result.
\end{proof}

\begin{proof}[Proof of Theorem \ref{thm:PositiveEquilibrium}] Let $K \subset \PbarS$ be a stoichiometric compatibility class. That is, there is a  composition $c^{\#} \in \PbarS$ such that
\begin{equation}
K = (c^{\#} + S) \cap \PbarS. \nonumber
\end{equation}
Because the network \rnet\,is conservative, $K$ is compact \cite{horn_general_1972}, and it is easily seen to be convex. Let \mbox{$g:K \to (c^{\#} + S)$} be defined by $g(c) := f(c)\,+\,c$, where $f(\cdot)$ is the species formation function for the kinetic system \kinsys. From Lemma \ref{lemma:LemmaForBrowder} it is apparent that for each composition $\bar{c}$ on the boundary of $K$ there is a number $\lambda$ and a composition $c^{\dagger} \in K$ (both depending perhaps on $\bar{c}$) such that $g(\bar{c}) - \bar{c} = \lambda(c^{\dagger} - \bar{c})$. In this case, it follows from a theorem of Browder \cite{Browder1967} that $g(\cdot)$ has a fixed point in $K$ --- that is, a composition $c^* \in K$ such that $g(c^*) = c^*$. But then we must have $f(c^*)= 0$, which is to say that there is an equilibrium in $K$.

	If the network is weakly reversible and concordant and if the stoichiometric compatibility class $K$ in nontrivial, it follows from Corollary \ref{corr:concordant_no-boundary-equil} that $c^*$ is not on the boundary of $K$; it must be positive. If, in addition, the kinetics is weakly monotonic, the fact that there cannot be another equilibrium in $K$ follows from Proposition \ref{prop:ConcordanceInjectivity}.
\end{proof}

\section{Concordance and properties of a network's fully open extension; eigenvalues}
\label{sec:FullyOpenSection}

	By the \emph{fully open extension} of reaction network \rnet we mean the network obtained from the original network by adjoining to it the species degradation reactions $\{s \to 0: s \in \scrS\}$ if those reactions are not already present. More precisely, the fully open extension of \rnet\; is the network $\{\scrS, \tilde{\scrC}, \tilde{\scrR} \}$, where\footnote{Note that in the definition of $\tilde{\scrC}$ we are viewing \scrS\  as a set in \RS --- that is, as a surrogate for the standard basis of \RS.  See Section \ref{sec:CRNTPreliminaries}. }
\begin{equation}
\label{eq:FullyOpenExtNotation}
\tilde{\scrC} := \scrC \cup \scrS \cup \{ 0 \}\quad \textrm{and} \quad \tilde{\scrR} := \scrR \cup \{s \to 0: s \in \scrS\}.
\end{equation}

	In this section we examine implications of concordance in a reaction network and in its fully open extension. We begin by noting that a network need not be concordant even when its fully open extension is concordant.  The simple network $B \leftarrow A \rightarrow C$ provides an example. 

\begin{rem}[\emph{Why study a network's fully open extension?}]

	There are two reasons to examine the concordance of a network's fully open extension. 	The first of these is that concordance of the fully open extension gives partial information about the \emph{stability} properties of equilibria for kinetic systems derived from the \emph{original} network, at least when the kinetics satisfies certain modest requirements. This we shall see in Theorem \ref{thm:EigenvalueThm} and in Section \ref{sec:Discordance}.

	The second reason is connected to inferences one can make from the original network's Species-Reaction Graph. As we shall see in a subsequent article, properties of a network's Species-Reaction Graph will often indicate that the \emph{fully open extension} of a network is concordant. For this reason, we will want to know when fully-open-network concordance implies concordance of the original network. This is the subject of Theorem \ref{thm:SmallToLargeConcordanceThm}, which also has some stability implications, discussed in Section \ref{sec:Discordance}.
\end{rem}

\medskip

	In a \emph{mass action} context, connections were made in \cite{craciun_multiple_2010} between injectivity of a reaction network and of its fully open extension. Here we will be able to generalize those results substantially. By way of preparation, we record the following definition:

\begin{definition}
\label{def:normal}
Consider a reaction network \rnet \  with stoichiometric subspace $S$. The network is $normal$ if there  are $q \in \PS$ and $\eta \in \PR$ such that the linear transformation $T: S \to S$ defined by
\begin{equation}
T \sigma := \sum_{y \to y' \in \scrR}\eta_{y \to y'}(y * \sigma)(y' - y)
\end{equation}
is nonsingular, where ``$*$" is the scalar product in \RS \  defined by
\begin{equation}
x*x' := \sum_{s \in \scrS}q_sx_sx'_s. 
\end{equation}
\end{definition}

\begin{rem}
The normality condition is not difficult to satisfy. In fact, it was shown in \cite{craciun_multiple_2010} that \emph{every weakly reversible network is normal}, regardless of the complexes. Weaker sufficient conditions for network normality are also given in \cite{craciun_multiple_2010}.
\end{rem}
\medskip

	We are now in a position to state a theorem that has considerable significance when placed alongside another one, to appear in a forthcoming article, that indicates how properties of a network's Species-Reaction Graph can ensure concordance of the network's fully open extension. In that theorem, however, concordance of the original network is left undecided. The following theorem completes the picture when the original network is normal. As we shall see, Theorem \ref{thm:SmallToLargeConcordanceThm} also has some stability consequences.  A proof is provided in Appendix C.

\begin{theorem}
\label{thm:SmallToLargeConcordanceThm}
A normal network is concordant if its fully open extension is concordant. In particular, a weakly reversible network is concordant if its fully open extension is concordant.
\end{theorem}
\medskip

	So far in this section all considerations were about the network alone, divorced from a kinetics. Now we will suppose that a network \rnet \, is endowed with a kinetics \scrK such that the resulting kinetic system \kinsys \, admits a positive equilibrium $c^*$.  We ask what network concordance ---  of either the network or of its fully open extension --- might tell us about eigenvalues of $df(c^*):S \to S$, where $S$ is the stoichiometric subspace for the network, $f(\cdot)$ is the species formation rate function for the kinetic system, and $df(c^*)$ is the derivative of $f(\cdot)$ at $c^*$.

	In this case, we will require that the kinetics be ``differentiably monotonic" at $c^*$. (It is easy to see that every mass action kinetics for a network is differentiably monotonic at every positive composition.)

\begin{definition} A kinetics \scrK \, for a reaction network \rnet\,is \emph{differentiably monotonic} at $c^* \in \PS$ if, for every  reaction $\rxn \in \scrR$, $\scrK_{\rxn}(\cdot)$ is differentiable at $c^*$ and, moreover, for each species $s \in \scrS$,
\begin{equation}
\frac{\partial \scrK_{\rxn}}{\partial c_s}(c^*) \ge 0,
\end{equation}
with inequality holding if and only if $s \in \supp y$. A \emph{differentiably monotonic kinetics} is one that is differentiably monotonic at every positive composition.
\end{definition}

\begin{theorem} 
\label{thm:EigenvalueThm}
Let \kinsys \  be a kinetic system with stoichiometric subspace $S$ and species formation rate function $f: \PbarS \to S$. Moreover, suppose that the kinetics is differentiably monotonic at $c^* \in \PS$. If the underlying network \rnet is concordant then the derivative $df(c^*): S \to S$ is nonsingular (whereupon 0 is not one of its eigenvalues). If the network's fully open extension is concordant then no real eigenvalue of  $df(c^*)$ is positive.
\end{theorem}

	Proof of Theorem \ref{thm:EigenvalueThm} is given in Appendix \ref{App:EigenvalueProofs}.

\begin{corollary}
\label{cor:NegativeEigenvalues}
Let \kinsys \  be a kinetic system with stoichiometric subspace $S$ and species formation rate function $f: \PbarS \to S$. Moreover, suppose that the kinetics is differentiably monotonic at $c^* \in \PS$. If  the underlying network \rnet is normal ---  in particular, if it is weakly reversible --- and if the network's fully open extension is concordant, then every real eigenvalue of $df(c^*)$ is negative.
\end{corollary}

\begin{proof} 
If the underlying network is normal then, by virtue of Theorem \ref{thm:SmallToLargeConcordanceThm},
its concordance follows from concordance of the fully open extension. In this case Theorem \ref{thm:EigenvalueThm} ensures that $df(c^*)$ is nonsingular, so $0$ cannot be among its eigenvalues. By the same theorem, concordance of the fully open extension then ensures that all real eigenvalues are negative. 
\end{proof}

\section{Consequences of discordance}
\label{sec:Discordance}

	Proposition \ref{prop:ConcordanceInjectivity} told us that a concordant network, taken with \emph{any} weakly monotonic kinetics, invariably gives rise to an injective kinetic system.  In this sense, concordance is a \emph{network} attribute that enforces injectivity for \emph{all} weakly monotonic kinetics. Proposition  \ref{prop:partialconverse} went even further: It told us that \emph{no discordant network has this same property}. 

	The observations in this section are in the same spirit, this time as counterpoints to Theorem \ref{thm:EigenvalueThm} and its corollary. Those told us that a normal network with concordant fully open extension, taken with \emph{any} differentiably monotonic kinetics, will invariably have some degree of stability at any positive equilibrium it might admit, in the sense that every real eigenvalue is negative. Here we shall see that no network with discordant fully open extension --- \emph{in particular, no weakly reversible discordant network} --- can have this same property (except in the trivial case for which there are no positive equilibria at all). 

	Because our interest will be in positive equilibria, our focus will be on networks with positively dependent reaction vectors. (Recall from Section \ref{sec:CRNTPreliminaries} that a kinetic system can admit a positive equilibrium only if the reaction vectors for the underlying network are positively dependent.) Proof of the following theorem is provided in Appendix \ref{App:EigenvalueProofs}.

\begin{theorem} 
\label{thm:DiscordanceEigenvalueThm}  Consider a reaction network with positively dependent reaction vectors. If the network is discordant, then there exists for it a differentiably monotonic kinetics such that the resulting kinetic system admits a positive \emph{degenerate} equilibrium. If the network's fully open extension is discordant then there exists for the original network a 
differentiably monotonic kinetics such that the resulting kinetic system 
admits an unstable positive equilibrium --- in fact, a positive equilibrium associated with a positive real eigenvalue.
\end{theorem}

\medskip
	Theorem \ref{thm:DiscordanceEigenvalueThm} has a number of corollaries. In each corollary, it is left implicit that the unstable positive equilibrium mentioned can be taken to have associated with it a positive real eigenvalue.

\begin{corollary}
\label{cor:DiscordantPositivelyDependent}
For every discordant normal network with positively dependent 
reaction vectors there exists a 
differentiably monotonic kinetics such that the resulting kinetic system 
admits a positive unstable equilibrium.
\end{corollary}

\begin{proof} It follows from Theorem \ref{thm:SmallToLargeConcordanceThm} that every normal discordant network has a discordant fully open extension. Corollary \ref{cor:DiscordantPositivelyDependent} then follows from Theorem \ref{thm:DiscordanceEigenvalueThm}.
\end{proof}

\begin{corollary}
\label{cor:DiscordantWeaklyReversible}
For every weakly reversible discordant network there is a 
differentiably monotonic kinetics such that the resulting kinetic system 
admits a positive unstable equilibrium.
\end{corollary}

\begin{proof}
A weakly reversible network is normal and, moreover, has positively dependent reaction vectors. Corollary \ref{cor:DiscordantWeaklyReversible} is then a consequence of Corollary \ref{cor:DiscordantPositivelyDependent}.
\end{proof}

	In preparation for the next corollary we point out that a concordant weakly reversible network (in fact, a reversible concordant network) can have a discordant fully open extension. Network \eqref{eq:ConcordantNetWithDiscordantExtension} provides an example, as indicated by computation\cite{Ji_toolboxWinV21}. (This is a variant of a network considered for different reasons in \cite{schlosser_theory_1994}.) 
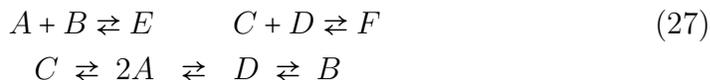
\begin{eqnarray}
\label{eq:ConcordantNetWithDiscordantExtension}
A + B \rightleftarrows E& &C + D \rightleftarrows F\\
C \; \rightleftarrows \;2A &\rightleftarrows& D \;\rightleftarrows \;B \nonumber
\end{eqnarray}\noindent

	The following is a consequence of Proposition \ref{prop:ConcordanceInjectivity}, Corollary \ref{corr:concordant_no-boundary-equil}, and Theorem \ref{thm:DiscordanceEigenvalueThm}.

\begin{corollary}
\label{cor:ConcordantWithDiscordantExtension}
For every weakly reversible concordant reaction network with discordant fully open extension there exists for the original network a differentiably monotonic kinetics such that the resulting kinetic system (for the original network) has the following properties: There is an equilibrium that is positive, unique within its stoichiometric compatibility class, and \emph{unstable}. Moreover, at each point on the boundary of that stoichiometric compatibility class the species-formation-rate vector points inward in the sense that there is an absent species produced at strictly positive rate.
\end{corollary}

\begin{rem}
\label{rem:NotAllRichBehaviorExcluded}
 Although network concordance, through Proposition \ref{prop:ConcordanceInjectivity}, ensures the uniqueness of a (positive) equilibrium within its stoichiometric compatibility class, network \eqref{eq:ConcordantNetWithDiscordantExtension} and   Corollary \ref{cor:ConcordantWithDiscordantExtension} tell us that concordance by itself cannot also ensure that the sole equilibrium is stable.  

	We will show in another article that when a normal (in particular, a weakly reversible) network's Species-Reaction Graph satisfies the conditions mentioned in Remark \ref{rem:SRGraphConditions}, concordance of \emph{both} the network \emph{and} its fully open extension follow. In this case, the hypothesis of Corollary \ref{cor:ConcordantWithDiscordantExtension} cannot be satisfied.

	Nevertheless, even when both a network and its fully open extension are concordant, and even when the kinetics is differentiably monotonic, we still cannot preclude instability of the unique positive equilibrium: Although its associated \emph{real} eigenvalues must be negative (Corollary \ref{cor:NegativeEigenvalues}), there might also be complex eigenvalues with positive real part.  Such eigenvalues were found to be extant in  Reidl et al. \cite{reidl_model_2006}, where network \eqref{eq:CalciumNet}, taken with a weakly monotonic (almost mass action) kinetics, served as a model for calcium oscillations in olfactory cilia.\footnote{We are grateful to Sayanti Banerjee for calling this paper to our attention.} As indicated by computation \cite{Ji_toolboxWinV21}, both network \eqref{eq:CalciumNet} and its fully open extension are concordant.

	Thus, it appears that concordance of a network (and of its fully open extension) enforce a \emph{degree} of dull behavior, at least to the extent that they preclude multiple stoichiometrically compatible equilibria (and positive real eigenvalues), but the capacity for other forms of interesting behavior might nevertheless remain. 
\end{rem}

\begin{rem}[Cyclic composition trajectories]
\label{rem:Bifurcations}
	Especially when the network in Corollary \ref{cor:ConcordantWithDiscordantExtension} is conservative, in which case stoichiometric compatibility classes are compact,  the existence of a (positive) equilibrium that is both unique within its stoichiometric compatibility class and also \emph{unstable} suggests that, for at least \emph{some} some differentiably monotonic kinetics, there might be an attracting periodic orbit.

	There are indeed compelling reasons to believe that, \emph{for any weakly reversible concordant network with discordant fully open extension, there will exist a differentiably monotonic kinetics such the resulting kinetic system gives rise to a cyclic composition trajectory}. In fact, certain considerations  point to the presence of a Hopf bifurcation in connection with kinetic variations: 

	To the weakly reversible network in Corollary \ref{cor:ConcordantWithDiscordantExtension} one can assign a complex balanced mass action kinetics \cite{horn_general_1972,horn_necessary_1972,feinberg_lectureschemical_1979}.  In this case the resulting mass action system will have the property that the unique equilibrium in each nontrivial stoichiometric compatibility class has associated with it eigenvalues\footnote{We are referring here to eigenvalues associated with eigenvectors in the stoichiometric subspace.} with negative real parts. On the other hand, Theorem \ref{thm:DiscordanceEigenvalueThm} ensures that, for the same network, there is another differentiably monotonic kinetics that gives rise to a nontrivial stoichiometric compatibility class containing a unique (positive) equilibrium associated with a positive real eigenvalue. Along a suitably parameterized transition from the first differentiably monotonic kinetics to the second, one therefore expects that a pair of complex-conjugate eigenvalues will cross the imaginary axis. (So long as the parameterized kinetics remains differentiably monotonic, it follows from Theorem \ref{thm:EigenvalueThm} that the crossing cannot be through zero.)

	Similar considerations apply even in the case of a weakly reversible concordant network with concordant fully open extension, provided that, for the original network, there is a differentiably monotonic kinetics for which there is a positive equilibrium having eigenvalues with positive real part. However, for a concordant network with concordant fully open extension there is no guarantee, of the kind  provided by Corollary \ref{cor:ConcordantWithDiscordantExtension}, that such a kinetics will exist.
\end{rem}

\section{Strong concordance and other generalizations}
\label{sec:generalizations}

	Motivated by a kinetic class considered in work by Banaji and Craciun, we introduce in this section the notion of \emph{strong concordance} and examine its consequences. Then we consider still other variants of the concordance idea.

\subsection{Strong concordance and two-way weakly monotonic kinetics}
\label{subsec:StrongConcordanceAndTwoWay}
	Some --- but certainly not all --- of what we have done so far relied on the supposition of \emph{weakly monotonic kinetics}, in which an increase in the rate of a particular reaction requires an increase in the concentration of at least one of its \emph{reactant} species. In particular, we argued in Proposition \ref{prop:ConcordanceInjectivity} that, when the kinetics is weakly monotonic, the species-formation-rate function for a concordant network is invariably injective and that, as a consequence, multiple stoichiometrically compatible equilibria, with at least one positive, are impossible. There was no requirement that there be a degradation reaction for every species.

	Although Banaji et al. \cite{banaji_PMatrix_2007} and Banaji and Craciun \cite{banaji_graph-theoretic_2010} \emph{do} require a complete supply of degradation reactions to get injectivity results, the kinetic requirements they invoke are, in most ways, significantly weaker. We shall be more specific about the meaning of what they call \emph{nonautocatalytic (NAC) kinetics}  in Remark \ref{rem:Banaji_CraciunDifficulties}, but it suffices here to say that, with NAC kinetics, the rate function for a particular reaction is required to have the property that an increase [decrease] in the concentration of one of its reactant [product] species --- keeping all other concentrations fixed --- cannot result in a decrease of the reaction rate.\footnote{Without much loss of generality, Banaji and Craciun presume that no reaction has a species common to both its reactant and product complexes.} (Species that, for the reaction, are neither reactants nor products are not permitted to have an effect on the rate.) When the stoichiometric matrix is ``strongly sign determined" --- i.e., every one of its square submatrices has a certain property ---  Banaji and Craciun are also able to assert injectivity and the absence of multiple steady states, but, again, only in the fully open setting. (In the fully open setting, the issue of stoichiometric compatibility becomes moot.)

	With weakly monotonic kinetics, then, the rate of a particular reaction can be influenced (positively) only by an \emph{increase} in the concentration of a \emph{reactant} species.  The ``influences" that NAC kinetics admits are broader: The rate of a reaction can \emph{also} be influenced (positively) by a \emph{decrease} in the concentration of a \emph{product} species. In rough terms, NAC kinetics admits ``two-way influences": It embraces product inhibition, while weakly monotonic kinetics does not. 

	For the most part, it is the purpose of this section to show that, without much difficulty, Proposition \ref{prop:ConcordanceInjectivity} can be generalized to subsume such two-way influences in the kinetics, provided that network concordance is replaced by \emph{strong concordance}.  We will assert that, when the kinetics is \emph{two-way weakly monotonic}, the species-formation-rate function for a \emph {strongly concordant} network is invariably injective and that, as a consequence, multiple stoichiometrically compatible equilibria, with at least one positive, are impossible. As was the case with Proposition \ref{prop:ConcordanceInjectivity} (and in contrast to the otherwise broad Banaji and Craciun results) there will be no requirement of a degradation reaction for every species.

	We begin by making precise what we mean by a two-way weakly monotonic kinetics. Hereafter, following Banaji and Craciun, \emph{we restrict our attention to networks having the property that no reaction has a species common to both its reactant and product complexes.}  In practical terms, this is a very small price to pay for the loss of ambiguity in speaking of a reactant or product species for a particular reaction.

	The following definition is intended to describe a kinetics in which an increase in the rate of a particular reaction requires that there be either an increase in the concentration of a reactant species or a decrease in the concentration of a product species. And for the rate to remain unchanged, it is necessary that there be no change in the concentrations of any reactant species or else that there be appropriately opposing changes in concentrations of the reactant and/or product species.

\begin{definition}
\label{DEF:TwoWayWeaklyMonotonic}
A kinetics $\mathscr{K}$ for a reaction network \rnet is \  \emph{two-way weakly monotonic}  if, for each pair of compositions $c^*$ and $c^{**}$, the following implications hold for each reaction $y \to y' \in \mathscr{R}$ such that $\supp y \subset \supp c^*$ and $\supp y \subset \supp c^{**}$:

\begin{enumerate}[(i)]
\item{$\mathscr{K}_{y \to y'}(c^{**}) > \mathscr{K}_{y \to y'}(c^*) \Rightarrow$ there is a species $s$ such that $\sgn (c^{**}_s - c^*_s) = \sgn (y - y')_s \neq 0$.}

\item{$\mathscr{K}_{y \to y'}(c^{**}) = \mathscr{K}_{y \to y'}(c^*) \Rightarrow$ $c^{**}_s = c^*_s$ for all $s \in \supp y$, or else there are species $s,s'$ with $\sgn (c^{**}_s - c^*_s) = \sgn (y - y')_s \neq 0$ and $\sgn (c^{**}_{s'} - c^*_{s'}) = - \sgn (y - y')_{s'} \neq 0$.}
\end{enumerate}
\end{definition} 
\begin{rem}
Definition \ref{DEF:TwoWayWeaklyMonotonic} provides for a kinetic class that \emph{subsumes} the class of weakly monotonic kinetics (Definition \ref{def:WeaklyMonotonic}). It is intended to permit, \emph{but not require}, the concentration of a reaction's product species to influence its rate. The terminology \emph{two-way weakly monotonic} should not be construed to suggest that there is a complete anti-symmetry with respect to the influence of reactant and product species:   Note that --- as with mass action kinetics --- the rate of a reaction can remain unchanged even when the concentrations of all of the reaction's \emph{product} species increases, the concentrations of all other species remaining fixed. The same is not true of the reaction's  \emph{reactant} species. An increase in the concentrations of all of its reactant species, the others remaining fixed, \emph{must} result in an increase in the reaction rate. Again, any kinetics that is weakly monotonic (e.g., mass action kinetics) is also two-way weakly monotonic.
\end{rem}

\begin{rem} In contrast to the definition of NAC kinetics \cite{banaji_PMatrix_2007,banaji_graph-theoretic_2010}, there is no requirement here that the rate functions in a two-way weakly monotonic kinetics be differentiable or even continuous. Although Definition \ref{DEF:TwoWayWeaklyMonotonic} is intended to capture the general spirit of NAC kinetics, there are also subtle, but consequential, differences that we shall discuss in Remark \ref{rem:Banaji_CraciunDifficulties}.
\end{rem}

	Because the class of two-way weakly monotonic kinetics is broader than the class of weakly monotonic kinetics, analogous theory will require focus on a narrower class of reaction networks --- in particular the class of networks we call \emph{strongly} concordant. In the following definition, the map $L$ is as in eq. \eqref{eq:LDefinition}.

\begin{definition}
\label{DEF:StronglyConcordant}
A reaction network \rnet with stoichiometric subspace $S$ is \emph{strongly concordant} if there do not exist $\alpha \in \ker L$ and a non-zero $\sigma \in S$ having the following properties:

\begin{enumerate}[(i)]
\item{For each $y \to y'$ such that $\alpha_{y \to y'} > 0$ there exists a species $s$ for which $\sgn \sigma_s = \sgn (y - y')_s \neq 0$.}
\item{For each $y \to y'$ such that $\alpha_{y \to y'} < 0$ there exists a species $s$ for which $\sgn \sigma_s = - \sgn (y - y')_s \neq 0$.}

\item{For each $y \to y'$ such that $\alpha_{y \to y'} = 0$, either (a) $\sigma_s = 0$ for all $s \in \supp y$, or (b) there exist species $s, s'$ for which $\sgn \sigma_s = \sgn (y - y')_s \neq 0$ and $\sgn \sigma_{s'} = -  \sgn (y - y')_{s'} \neq 0$.}
\end{enumerate}
\end{definition}

	Note that a network that is strongly concordant is also concordant. Network
\eqref{eq:SchlosserExampleClosed} is concordant but not strongly concordant.
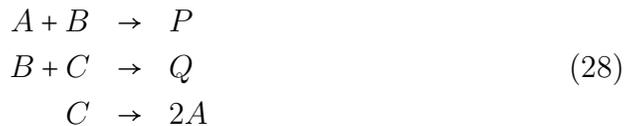
\begin{eqnarray}
\label{eq:SchlosserExampleClosed}
A + B &\to& P \nonumber \\
B + C &\to& Q \\
C &\to& 2A \nonumber
\end{eqnarray}

\begin{rem} [Ascertaining strong concordance] We intend to show in a separate article that strong concordance of a fully open network can be asserted when its Species-Reaction Graph satisfies the conditions stated in Remark \ref{rem:SRGraphConditions}. More generally, strong concordance of a network, not necessarily fully open, can be determined, either positively or negatively, by a sign-checking procedure implemented in \cite{Ji_toolboxWinV21}. 
\end{rem}
\smallskip
	The following proposition generalizes Proposition \ref{prop:ConcordanceInjectivity}.

\begin{proposition}[Generalization of Proposition \ref{prop:ConcordanceInjectivity}]
\label{prop:StrongConcordanceInjectivity}
A two-way weakly monotonic kinetic system \kinsys\  is injective whenever its underlying reaction network \rnet\  is strongly concordant. In particular, if the underlying reaction network is strongly concordant, then the kinetic system cannot admit two distinct stoichiometrically compatible equilibria, at least one of which is positive.
\end{proposition}

	The proof of the proposition closely parallels the fairly simple proof of Proposition \ref{prop:ConcordanceInjectivity}. In fact, Proposition \ref{prop:StrongConcordanceInjectivity} is merely a special case of a more general proposition that we state and prove in \S\ref{subsec:StillMoreGeneral}, one that allows for very general species-influences in the kinetics. 

	Before we do that, however, we want to emphasize once again that --- in contrast to the work of Banaji and Cracun  --- Proposition \ref{prop:StrongConcordanceInjectivity} does not require the fully open setting. That they were compelled to require a full supply of degradation reactions in order to ensure injectivity of the species-formation-rate function is due, in part, to difficulties inherent in the NAC kinetic class. Although NAC kinetics is close in spirit to what we call two-way weakly monotonic kinetics, there are subtle distinctions between the two that, for NAC kinetics in the true chemical reactions, makes a result resembling Proposition \ref{prop:StrongConcordanceInjectivity} impossible, at least in the absence of additional special reactions (e.g., degradation reactions) conforming to stronger kinetic requirements. This we explain in the following remark.

\begin{rem}[Difficulties instrinsic to NAC kinetics] 
\label{rem:Banaji_CraciunDifficulties}
In the sense of \mbox{\cite{banaji_PMatrix_2007,banaji_graph-theoretic_2010},} a (differentiable) kinetics $\mathscr{K}$ for a reaction network \rnet is \emph{non-autocatalytic} (NAC) if for each $c^* \in \PS$, for each reaction $y \to y' \in \mathscr{R}$, and for each species $s \in \mathscr{S}$, the following conditions hold:\begin{enumerate}
\item{$\frac{\partial \mathscr{K}_{y \to y'}}{\partial{c_s}}(c^*)(y'_s - y_s) \leq 0$},
\item{$y'_s - y_s = 0 \Rightarrow \frac{\partial\mathscr{K}_{y \to y'}}{\partial{c_s}}(c^*) = 0$.}
\end{enumerate}

	For \emph{any} reaction network \rnet, it is not difficult to imagine an NAC kinetic system \kinsys \  in which, for some open set $\Omega \in \PS$, each rate function $\scrK_{\rxn}(\cdot)$ is \emph{constant} on $\Omega$. In this case, the species-formation-rate function is also constant on $\Omega$, so the kinetic system is \emph{not} injective. Thus, there can be no assertion resembling Proposition \ref{prop:StrongConcordanceInjectivity}, in which every  reaction network within a certain class, when taken with \emph{every} NAC kinetics, invariably gives rise to an injective kinetic system.

	In the work of Banaji and Craciun injectivity results from the supposition that the ``true" chemical reactions --- those that conform to NAC kinetics --- are supplemented by degradation reactions that conform to still stronger kinetic requirements.
\end{rem}
\medskip
	The following theorem is a variant of Theorem \ref{thm:SmallToLargeConcordanceThm}. It will primarily find use in the following way: We intend to show in a separate paper that, when a reaction network's Species-Reaction Graph satisfies the conditions of Remark \ref{rem:SRGraphConditions}, then the network's fully open extension is strongly concordant. Thus, when the network is normal, those same Species-Reaction Graph conditions will ensure strong concordance of the original network.

\begin{theorem}
\label{thm:SmallToLargeStrongConcordanceThm}
A normal network is strongly concordant if its fully open extension is strongly concordant. In particular, a weakly reversible network is strongly concordant if its fully open extension is strongly concordant.
\end{theorem}
\medskip
	Proof of the theorem is very similar to the proof of Theorem \ref{thm:SmallToLargeConcordanceThm}. See Appendix \ref{App:ProofOfSmallToLargeConcordanceThm}.

	 Theorem \ref{thm:EigenvalueThm} also has a generalization. In preparation for it we need the following definition:

\begin{definition}
\label{DEF:DifferentiablyTwoWay}
A kinetics $\mathscr{K}$ for a reaction network \rnet\  is \emph{differentiably two-way monotonic} at $c^* \in \PS$ if, for every reaction $y \to y' \in \mathscr{R}$, $\mathscr{K}_{y \to y'}(\cdot)$ is differentiable at $c^*$ and, moreover,

\begin{enumerate}[(i)]
\item{$\frac{\partial \mathscr{K}_{y \to y'}}{\partial c_s}(c^*) > 0 \hspace{0.5 cm} \forall s \in \supp y$,} 
\item{$\frac{\partial \mathscr{K}_{y \to y'}}{\partial c_s}(c^*) \leq 0 \hspace{0.5 cm} \forall s \in \supp y'$,} 
\item {$\frac{\partial \mathscr{K}_{y \to y'}}{\partial c_s}(c^*) = 0 \hspace{0.5 cm} \forall s \in \mathscr{S} \setminus (\supp y \cup \supp y')$.}
\end{enumerate}
\end{definition}

	Proofs of the following theorem and its corollary proceed very much like the proofs of Theorem \ref{thm:EigenvalueThm} and its corollary.

\begin{theorem} 
\label{thm:EigenvalueThmStrongConc}
Let \kinsys \  be a kinetic system with stoichiometric subspace $S$ and species formation rate function $f: \PbarS \to S$. Moreover, suppose that the kinetics is differentiably two-way monotonic at $c^* \in \PS$. If the underlying network \rnet is strongly concordant then the derivative $df(c^*): S \to S$ is nonsingular (whereupon 0 is not one of its eigenvalues). If the network's fully open extension is strongly concordant then no real eigenvalue of  $df(c^*)$ is positive.
\end{theorem}

\begin{corollary}
Let \kinsys \  be a kinetic system with stoichiometric subspace $S$ and species formation rate function $f:\PbarS \to S$. Moreover, suppose that the kinetics is differentiably two-way monotonic at $c^* \in \PS$. If  the underlying network \rnet \  is normal ---  in particular, if it is weakly reversible --- and if the network's fully open extension is strongly concordant, then every real eigenvalue of $df(c^*)$ is negative.
\end{corollary}

\subsection{Still more general species influences on reaction rates}
\label{subsec:StillMoreGeneral}

	In our definition of \emph{two-way weakly monotonic kinetics}, we generalized \emph{weakly monotonic kinetics} to permit (but not require) for a particular reaction an \emph{inhibitory} effect exerted by the reaction's \emph{product} species. That is, a particular kinetics within  the two-way weakly monotonic class  was permitted to have the property that the rate of a particular reaction \rxn \  might increase whenever the concentration of $s \in \supp y'$ decreases, the concentrations of all other species remaining fixed. The companion notion of  \emph{strong concordance} --- an attribute of network structure alone, divorced from kinetic considerations --- worked hand-in-hand with the broader kinetic class to give rise to the injectivity result given in Proposition \ref{prop:StrongConcordanceInjectivity}.

	It is not substantially more difficult to consider, in an analogous way, kinetic classes with far wider ranges of permissible species influences. Indeed, for a particular reaction we might want to permit inducer or inhibitor effects on the reaction rate by species that, for the reaction, are neither reactants nor products.  (We shall continue to require that, for each reaction, an increase in the concentration of one its reactant species, with all other species concentrations held constant, results in an increase of the reaction's rate.) 

	We begin by indicating formally what we mean by an \emph{influence specification} for a network. This amounts to indicating, for each reaction, which species are deemed to be possible inducers and which are deemed  to be possible inhibitors (and which species are required to have no effect whatsoever).

\begin{definition}
\label{DEF:Generalization}
An \emph{influence specification} \scrI\  for a reaction network \rnet\  is an assignment to each reaction \rxn\  of a function $\scrI_{\rxn}: \scrS \to \{1,0,-1\}$ such that
\begin{equation}
\label{DEF:InfluenceInclusion}
\scrI_{y \to y'}(s)  = 1, \quad \forall s \in \supp y.
\end{equation}
If $\scrI_{y \to y'}(s)$ = 1 [resp., -1] then species $s$ is an \emph{inducer} [\emph{inhibitor}] of reaction \rxn.
\end{definition}

	Recall that, for the two-way monotonic kinetic class, we permitted, but did not \emph{require},  the product species to be inhibitors for a particular kinetics within the class. In the following definition, our intent is to specify, for a broad kinetic class, which species are permitted to be inhibitors or inducers --- without \emph{requiring} them to be one or the other --- for a particular kinetics within the class. (The exception is that, in the sense of Definition \ref{DEF:Generalization}, we require that, for each  reaction, its reactant species are inducers.)

\begin{definition}
\label{DEF:WeaklyMonotonicWithRespect}
A kinetics $\mathscr{K}$ for a reaction network \rnet is \emph{weakly monotonic 
with respect to influence specification} $\mathscr{I}$ if, for every pair of compositions $c^*$ and $c^{**}$, the following implications hold for each reaction $y \to y' \in \mathscr{R}$ such that $\supp y \subset \supp c^*$ and $\supp y \subset \supp c^{**}$:

\begin{enumerate}[(i)]
\item{$\mathscr{K}_{y \to y'}(c^{**}) > \mathscr{K}_{y \to y'}(c^*) \Rightarrow$ there is a species $s$ such that $\sgn (c^{**}_s - c^*_s) =  \scrI_{y \to y'}(s) \neq 0$.}
\item{$\mathscr{K}_{y \to y'}(c^{**}) = \mathscr{K}_{y \to y'}(c^*) \Rightarrow$ either (a) $c^{**}_s = c^*_s$ for all $s \in \supp y$ or (b) there are species $s,s'$ with $\sgn (c^{**}_s - c^*_s) = \scrI_{y \to y'}(s) \neq 0$ and \mbox{$\sgn (c^{**}_{s'} - c^*_{s'})$} = -  $\scrI_{y \to y'}(s') \neq 0$.}
\end{enumerate}
\end{definition}

\begin{rem}
\label{REM:Boundary}
Note that if a kinetics $\mathscr{K}$ is weakly monotonic with respect to influence specification $\mathscr{I}$ and if $c^{**}$ is a positive composition then implications $(i)$ and $(ii)$ in Definition \ref{DEF:WeaklyMonotonicWithRespect} will hold for all reactions in the network, even when $c^{*}$ is not strictly positive. 
\end{rem}

	In the following definition, the map $L$ is again as in eq. \eqref{eq:LDefinition}.

\begin{definition}
\label{DEF:ConcordantWRInfluence}
A reaction network \rnet\  with stoichiometric subspace $S$ is \emph{concordant with respect to influence specification} \scrI\  if there do not exist $\alpha \in \ker L$ and a non-zero $\sigma \in S$ having the following properties:

\begin{enumerate}[(i)]
\item{For each $y \to y'$ such that $\alpha_{y \to y'} > 0$ there exists a species $s$ for which $\sgn \sigma_s =  \scrI_{y \to y'}(s) \neq 0$.}
\item{For each $y \to y'$ such that $\alpha_{y \to y'} < 0$ there exists a species $s$ for which $\sgn \sigma_s = -  \scrI_{y \to y'}(s) \neq 0$.}
\item{For each $y \to y'$ such that $\alpha_{y \to y'} = 0$,} either 
\smallskip
(a) $\sigma_s = 0$ for all $s \in \supp y$ or 
(b) there are species  $s, s'$ for which $\sgn \sigma_s =  \scrI_{y \to y'}(s) \neq 0$ and $\sgn \sigma_{s'} = -  \scrI_{y \to y'}(s') \neq 0$.

\end{enumerate}

\end{definition}

\begin{rem} It is expected that \cite{Ji_toolboxWinV21} will be expanded to ascertain not only concordance and strong concordance but also the more general form of concordance indicated in Definition \ref{DEF:ConcordantWRInfluence}.
\end{rem}

\begin{rem} 
\label{rem:HowOtherDefsEmerge}
The definitions of the weakly monotonic kinetic class (Definition \ref{def:WeaklyMonotonic}) and of ordinary network concordance (Definition \ref{def:Concordant}) emerge as special cases of Definitions \ref{DEF:WeaklyMonotonicWithRespect} and \ref{DEF:ConcordantWRInfluence} applied with the influence specification given, for each reaction, by $\scrI_{\rxn} = \sgn y$.  Similarly, the definitions of the two-way monotonic class and of strong concordance emerge from the influence specification given, for each reaction, by $\scrI_{\rxn} = \sgn (y - y')$.
\end{rem}

	The following proposition generalizes Proposition \ref{prop:StrongConcordanceInjectivity}. In light of Remark \ref{rem:HowOtherDefsEmerge}, its proof will serve also as a proof of Proposition \ref{prop:StrongConcordanceInjectivity}.

\begin{proposition}[Generalization of Proposition \ref{prop:StrongConcordanceInjectivity}]
\label{PROP:Generalization}
A kinetic system \\ \kinsys \  is injective whenever there exists an influence specification $\mathscr{I}$ such that:
\begin{enumerate}[(i)]
\item{The kinetics $\mathscr{K}$ is weakly monotonic with respect to $\mathscr{I}$.}
\item{The underlying network \rnet\  is concordant with respect to $\mathscr{I}$.}
\end{enumerate}
\end{proposition}

\begin{proof}
Let \kinsys \  be a kinetic system and suppose that $\mathscr{I}$ is an influence specification that satisfies conditions $(i)$ and $(ii)$ in the proposition statement.  Suppose also that, contrary to what is to be proved, there are distinct stoichiometrically-compatible compositions $c^*$ and $c^{**}$, with \mbox{$c^{**} \in \PS$,} such that

\begin{equation}
\label{EQ:NonInjectivityCondition}
\sum_{y \to y' \in \mathscr{R}} \mathscr{K}_{y \to y'}(c^{**})(y' - y) = \sum_{y \to y' \in \mathscr{R}} \mathscr{K}_{y \to y'}(c^{*})(y' - y).
\end{equation}

\noindent 
Then  $\alpha \in \RR$, defined by

\begin{equation}
\label{EQ:AlphaInj}
\alpha_{y \to y'} := \mathscr{K}_{y \to y'}(c^{**}) - \mathscr{K}_{y \to y'}(c^*), \hspace{5 mm} \forall y \to y' \in \mathscr{R},
\end{equation}

\noindent is a member of $\ker L$. The vector $\sigma := c^{**} - c^*$ is clearly a member of the stoichiometric subspace, $S$.

Note that, if for a particular reaction $y \to y'$ we have $\alpha_{y \to y'} > 0$, then weak monotonicity with respect to $\mathscr{I}$ implies that, for some species $s$, $\sgn \sigma_s = \sgn (c^{**}_s - c^*_s) =  \scrI_{y \to y'}(s)  \neq 0$. Similarly, when $\alpha_{y \to y'} < 0$, there is a species $s$ such that $\sgn \sigma_s = \sgn (c^{**}_s - c^*_s) = -\sgn (c^*_s - c^{**}_s) =  -  \scrI_{y \to y'}(s) \neq 0$. Finally, if $\alpha_{y \to y'} = 0$, then $\sigma_s = c^{**}_s - c^*_s = 0$ for all $s \in \supp y
$, or else there exist distinct species $s, s'$ such that $\sgn \sigma_s = \sgn (c^{**}_s - c^*_s) =  \scrI_{y \to y'}(s) \neq 0$ and $\sgn \sigma_{s'} = \sgn (c^{**}_{s'} - c^*_{s'}) = -  \scrI_{y \to y'}(s') \neq 0$.

The existence of $\alpha$ and $\sigma$ so constructed contradicts the supposition that the network \rnet is concordant with respect to influence specification $\mathscr{I}$.
\end{proof}
\smallskip
	It is not difficult to see that other propositions and theorems in this paper can also be generalized in various ways to accommodate the broader notions of concordance and kinetic monotonicity, both relative to a common influence specification.
\smallskip
\begin{rem}[Kinetic classes with \emph{mandatory} inducer-inhibitor influences] For a reaction network \rnet, consider the class of all kinetics that are weakly monotonic with respect to a certain influence specification, \scrI. Moreover, let \rxn\  be a fixed reaction, and suppose that species $s$ is such that $\scrI_{y \to y'}(s) < 0$. (This choice of sign is made only for the sake of concretion.)

	Now if \scrK\  is a fixed kinetics in the kinetic class under consideration, then, with respect to \scrK, species $s$ might have no effect at all on the rate of reaction \rxn. With another kinetics $\scrK^\prime$ in the same class, $s$ might indeed be a nontrivial inhibitor of \rxn. That  $\scrI_{y \to y'}(s)$\  is negative merely indicates that the kinetic class under consideration is sufficiently wide as to admit \emph{some} kinetics, such as $\scrK^\prime$, in which $s$ is an inhibitor of \rxn.

	In this way, mass action kinetics (for which there are no inhibitors) can be viewed to sit within a larger class of kinetics, certain members of which might include inhibition. Theorems can then be stated for broad classes of kinetics that include mass action kinetics as a special subclass.

	We could just as well have chosen to consider kinetic classes within which \emph{every} kinetics is \emph{required} to include, for example, a nontrivial inhibition of species $s$ on reaction \rxn. Specification of such compulsory influences can be accomplished in the following way: Definition \ref{DEF:Generalization} would be expanded to include, for each reaction \rxn, a specified species set $\scrM_{\rxn}  \subset \scrS$, with $\supp y \subset \scrM_{\rxn}$, for which the influences given by $\scrI_{\rxn}$ are mandatory.

	In item (ii)(a) of  Definition \ref{DEF:WeaklyMonotonicWithRespect}, $\supp y$ would then be replaced by $\scrM_{\rxn}$, and the same replacement would be made in item (iii)(a) of Definition \ref{DEF:ConcordantWRInfluence}.  With very minor adjustments in the proof, Proposition \ref{PROP:Generalization} would emerge as before.
\end{rem}

\section*{Acknowledgments} We are grateful to Uri Alon for his encouragement and support and to Avi Mayo for his help. We especially thank Haixia Ji and Daniel Knight for expanding \emph{The Chemical Reaction Network Toolbox} \cite{Ji_toolboxWinV21} to implement concordance tests.

\appendix
\numberwithin{equation}{section}
\appendixpage
\section{Proof of Proposition \ref{prop:partialconverse}.}
Here we prove the following proposition:

\medskip
\noindent
\textbf{Proposition \ref{prop:partialconverse}.} \emph{Let \rnet\  be a reaction network that is discordant. Then there exists for it a weakly monotonic kinetics \scrK \ such that the resulting kinetic system \kinsys\ is not injective.}

\begin{proof} Because the network \rnet\  is discordant there exist for it a nonzero $\sigma$ in the stoichiometric subspace, $S$, and an $\alpha \in \RR$ that together satisfy the conditions in Definition \ref{def:Concordant}. In particular we have
\begin{equation}
\label{EQ:NonInj1Again}
\sum_{y \to y' \in \mathscr{R}} \alpha_{y \to y'} (y' - y) = 0.
\end{equation}

	From the conditions in Definition \ref{def:Concordant} it is not difficult to see that, for each $\rxn \in \scrR$, there is a vector $p_{\rxn} \in \PbarS$ such that 

\begin{equation}
\supp p_{\rxn} =  \supp y \quad \mathrm{and} \quad \alpha_{\rxn} = p_{\rxn} \cdot \sigma.
\end{equation}
\smallskip

	Because $\sgn x = \sgn(\exp(x) - 1)$ for each real $x$, there are positive numbers $\{ \eta_{\rxn} \}_{\rxn \in \scrR}$  such that \eqref{EQ:NonInj1Again} can be rewritten as
\begin{equation}
\label{EQ:NonInj1Again2}
\sum_{\rxn \in \scrR} \eta_{\rxn} \left[\exp(p_{\rxn} \cdot \sigma) - 1\right](y' - y) = 0.
\end{equation}For the same reason, there exists $c^* \in \PS$ such that 
\begin{equation}
\label{EQ:cStar} 
\sigma_s = c_s^*\left[\exp(\sigma_s) - 1 \right], \quad \forall s \in \scrS.
\end{equation}

\noindent We define $c^{**} \in \PS$ by

\begin{equation}
\label{EQ:cDoubleStar} \\ 
c^{**} := c^* \circ e^\sigma.
\end{equation}

\noindent Then
\begin{equation}
c^{**} - c^* = \sigma \in S
\end{equation}
and
\begin{equation}
\label{eq:cRatio}
\frac{c^{**}}{c^*} = e^{\sigma}.
\end{equation}

\medskip
\noindent
From \eqref{EQ:NonInj1Again2} and \eqref{eq:cRatio} it follows that
\begin{equation}
\label{eq:Pre-Noninjectivity}
\sum_{\rxn \in \scrR}\eta_{\rxn}\left[ 
{\left( \frac{c^{**}}{c^*}\right)}^{p_{\rxn}} - 1  \right](y' - y) = 0.
\end{equation}
	Now we construct a weakly monotonic kinetics  for the network in the following way: For each $\rxn \in \scrR$ we let
\begin{equation}
\label{eq:KineticsConstructed}
\scrK_{\rxn}(c) := \eta_{\rxn} {\left( \frac{c}{c^*}\right)}^{p_{\rxn}}.
\end{equation}
From \eqref{eq:Pre-Noninjectivity} and \eqref{eq:KineticsConstructed} it follows that
\begin{eqnarray}
\sum_{\rxn \in \scrR}\scrK_{\rxn}(c^{**})(y' - y) = \sum_{\rxn \in \scrR}\scrK_{\rxn}(c^{*})(y' - y),
\end{eqnarray}
whereupon the kinetic system \kinsys \ is not injective.

\end{proof}

\section{An alternative proof of Theorem \ref{thm:concordant_reaction-transitive}.}	
\label{App: AlternativeProofOfBdryThm}
In the proof of Theorem \ref{thm:concordant_reaction-transitive} offered earlier we required that every mass action system for which the underlying reaction network is weakly reversible admits a strictly positive equilibrium in each nontrivial stoichiometric compatibility class, regardless of the (positive) values that the rate constants take. That this is so is proved in a currently unpublished manuscript \cite{deng_steadystates}. With this in mind, we offer in this appendix a slightly different proof which instead draws on established properties of complex balanced mass action systems. (See, in particular, \cite{horn_general_1972}. Here we  draw heavily on \cite{feinberg_lectureschemical_1979}.)

We begin with a lemma:

\begin{lemma}
\label{lem:CBLemma}
Let \{$\scrS, \scrC ^{\dag}, \scrR ^{\dag}$\} be a weakly reversible subnetwork of reaction network \rnet, and let $c \in \PS$ be a fixed but arbitrary positive composition. Then it is possible to find positive numbers (rate constants) $\{k_{\rxn}\}_{\rxn \in \scrR ^{\dag}}$ that satisfy
\begin{equation}
\label{eq:SubnetworkEquilEqn}
\sum_{\rxn \in \scrR ^{\dag}}k_{\rxn}(c)^y(y' - y) = 0.
\end{equation}
In fact, $\{k_{\rxn}\}_{\rxn \in \scrR ^{\dag}}$ can be chosen to satisfy the stronger requirement
 \begin{equation}
\label{eq:SubnetworkCBEqn}
\sum_{\rxn \in \scrR ^{\dag}}k_{\rxn}(c)^y(\omega_{y'} - \omega_y) = 0,
\end{equation}
where $\{\omega_y\}_{y \in \scrC}$ is the standard basis for \RC.
\end{lemma}

\begin{proof} (See in particular Lecture 4 in \cite{feinberg_lectureschemical_1979}.) Because \{$\scrS, \scrC ^{\dag}, \scrR ^{\dag}$\} is weakly reversible, there exist positive numbers $\{\kappa_{\rxn}\}_{\rxn \in \scrR ^{\dag}}$ that satisfy
 \begin{equation}
\sum_{\rxn \in \scrR ^{\dag}}\kappa_{\rxn}(\omega_{y'} - \omega_y) = 0,
\end{equation}
and, therefore, 
\begin{equation}
\sum_{\rxn \in \scrR ^{\dag}}\kappa_{\rxn}(y' - y) = 0.
\end{equation}
Now choose  $\{k_{\rxn}\}_{\rxn \in \scrR ^{\dag}}$ to satisfy
\begin{equation}
\label{eq:RateConstChoice3}
\kappa_{\rxn} = k_{\rxn}c^y, \quad \forall \rxn \in \scrR ^{\dag}.
\end{equation}
\end{proof}

\begin{rem}
\label{rem:CBRemark}
When the rate constants have been chosen to satisfy condition \eqref{eq:SubnetworkEquilEqn}, the composition $c$ becomes an equilibrium of the mass action systems $\{\scrS, \scrC ^{\dag}, \scrR ^{\dag}, k \}$. When the rate constants satisfy the stronger condition \eqref{eq:SubnetworkCBEqn}, then $c$ is said to be an equilibrium at which \emph{complex balancing} obtains, and  $\{\scrS, \scrC ^{\dag}, \scrR ^{\dag}, k \}$ is said to be a \emph{complex balanced mass action system}. Properties of complex balanced mass action systems are discussed in \cite{horn_general_1972} and \cite{feinberg_lectureschemical_1979}.
\end{rem}
\medskip
	With this as background we state Theorem \ref{thm:concordant_reaction-transitive} once again and then provide an alternate proof.
\medskip

\noindent
\textbf{Theorem \ref{thm:concordant_reaction-transitive}} \emph{
For a weakly reversible concordant reaction network, no nontrivial stoichiometric compatibility class can have on its boundary a composition that is reaction-transitive.}

\begin{proof}[Alternative proof of Theorem \ref{thm:concordant_reaction-transitive}] Suppose that \rnet is a concordant weakly reversible reaction network and that, contrary to what is to be proved, $c^* \in \PbarS$ is a reaction-transitive composition on the boundary of a nontrivial stoichiometric compatibility class. Let
\begin{eqnarray}
\scrS ^* &:=& \supp c^*, \nonumber\\
\scrC ^* &:=& \{y \in \scrC : \supp y \subset \supp c^*\}\\ 
\scrR ^* &:=& \{ y \to y' \in \scrR : \supp y \subset \supp c^*\}. \nonumber
\end{eqnarray}

To be concrete, we shall begin by considering the more difficult situation in which $\scrC ^*$ is not empty.  By virtue of Remark  \ref{rem:PrepForBoundaryThm},  \{$\scrS ^*, \scrC ^*, \scrR ^*$\} is a weakly reversible subnetwork of the original network. 

	Now we choose $c^{\#}$ to be a strictly positive composition in  \RS\  that agrees with $c^*$ for every species in $\scrS ^*$ --- that is  $c_s^{\#} =  c_s^*$  for all $s \in \scrS ^*$. By virtue of Lemma \ref{lem:CBLemma} and Remark \ref{rem:CBRemark} we can choose positive rate constants $\{k^*_{\rxn}\}_{\rxn \in \scrR ^*}$ to ensure that  $c^{\#}$ is a complex balanced equilibrium of the mass action system \{$\scrS, \scrC ^*, \scrR ^*, k^*$\}. Because $c^{\#}$  agrees with $c^*$ on $\scrS ^*$ and because every complex in $\scrC ^*$ has support in  $\scrS ^*$ we have
\begin{equation}
\label{eq:c*EquilEqn}
\sum_{\rxn \in \scrR ^*}k^*_{\rxn}(c^*)^y(y' - y) = \sum_{\rxn \in \scrR ^*}k^*_{\rxn}(c^{\#})^y(y' - y) = 0.
\end{equation}
(Note that species in $\scrS \setminus \scrS ^*$ are ``inerts" relative to reactions in $\scrR ^*$; they are neither consumed nor produced by those reactions.)

	Because \{$\scrS, \scrC ^*, \scrR ^*, k^*$\} is a complex balanced mass action system, it has the properties that accrue to such a system \cite{horn_general_1972,feinberg_lectureschemical_1979}. Among these is the existence of precisely one positive equilibrium in each nontrivial stoichiometric compatibility class \emph{for the underlying network}  $\{\scrS, \scrC ^*, \scrR ^*\}$. Note that a stoichiometric compatibility class for that network  is the intersection of \PbarS with a  parallel of its stoichiometric subspace,
\begin{equation}
S^* := \mathrm{span} \{y' - y \in \RS: \rxn \in \scrR ^*\}. \nonumber
\end{equation}
On the other hand, a stoichiometric compatibility class of the parent network \rnet is the intersection of \PbarS with a parallel of \emph{its}  stoichiometric subspace,
\begin{equation}
S := \mathrm{span} \{y' - y \in \RS: \rxn \in \scrR\}, \nonumber
\end{equation}
\noindent Note that $S$ \emph{contains} $S^*$.

	In fact, relative to \rnet, the nontrivial stoichiometric compatibility class $(c^* + S) \cap \PbarS$ containing the boundary composition $c^*$ is the union of stoichiometric compatibility classes for the network $\{\scrS, \scrC ^*, \scrR ^*\}$. Clearly, some of these contain positive compositions; those that do are nontrivial, and each therefore contain precisely one positive equilibrium \emph{for the mass action system}  \{$\scrS, \scrC ^*, \scrR ^*, k^*$\}. 

		Hereafter we suppose that $c^{**}$ is one such positive equilibrium for  the mass action system \{$\scrS, \scrC ^*, \scrR ^*, k^*$\}, residing in $(c^* + S) \cap \PbarS$; thus, we have 

\begin{equation}
\label{eq:c**Eq1}
\sum_{\rxn \in \scrR ^*}k^*_{\rxn}(c^{**})^y(y' - y) = 0.
\end{equation}

	Having already chosen rate constants for reactions in \scrR*, we can now choose rate constants for the remaining reactions --- that is, those in $\scrR \setminus \scrR ^*$ --- to ensure that $c^{**}$ is an equilibrium of the resulting larger mass action system: In fact, because the subnetwork $\{\scrS, \scrC \setminus \scrC^*, \scrR \setminus \scrR ^*\}$ is also weakly reversible (Remark \ref{rem:PrepForBoundaryThm}), Lemma \ref{lem:CBLemma} ensures that we can choose positive numbers  $\{k^{**}_{\rxn}\}_{\rxn \in  \scrR \setminus \scrR ^*}$ to satisfy
\begin{equation}
\label{eq:c**Eq2}
\sum_{\rxn \in \scrR \setminus \scrR ^*}k^{**}_{\rxn}(c^{**})^y(y' - y) = 0.
\end{equation}
Finally, we denote by $\bar{k} \in \PR$ the rate constant assignment for the full network constructed as indicated --- i.e., $\bar{k}_{\rxn} = k^*_{\rxn}\; \textrm{  for all  } \; \rxn \in \scrR ^* \textrm{ and } \bar{k}_{\rxn} = k^{**}_{\rxn} \;\textrm{ for all } \; \rxn \in \scrR \setminus \scrR ^*$.

	Taken together, \eqref{eq:c**Eq1} and \eqref{eq:c**Eq2} ensure that $c^{**}$ is a positive equilibrium of the mass action system $\{\scrS,\scrC,\scrR,\bar{k}\}$. Moreover, $c^*$ is a (boundary) equilibrium of that same mass action system: This follows from \eqref{eq:c*EquilEqn} and the fact that, for each $\rxn \in \scrR \setminus \scrR ^*,\; \supp y \not\subset \supp c^*$. 

	Because, relative to the concordant network \rnet,  the distinct equilibria   $c^{*}$ and  $c^{**}$ are stoichiometrically compatible, we have a contradiction of Proposition \ref{prop:ConcordanceInjectivity}.

	We must still consider the case in which  $\scrR ^*$ is empty.  In this case,  $\supp y \not\subset \supp c^*$ for all $y \in \scrC$, so, for any kinetics assigned to \rnet, $c^*$ will be an equilibrium of the resulting kinetic system.  Now let $c^{**}$ be a positive composition in the nontrivial stoichiometric compatibility class containing  $c^{*}$. We can, by virtue of Lemma \ref{lem:CBLemma}, assign a positive rate constant to each reaction of \scrR \, in such a way as to make  $c^{**}$ an equilibrium of the resulting mass action system. Then $c^{*}$ and $c^{**}$ are distinct stoichiometrically compatible equilibria of that same mass action system. Because the network \rnet\  is concordant, we again have a contradiction of  Proposition \ref{prop:ConcordanceInjectivity}. 
\end{proof}

\section{Proof of Theorems \ref{thm:SmallToLargeConcordanceThm} and \ref{thm:SmallToLargeStrongConcordanceThm}.} 
\label{App:ProofOfSmallToLargeConcordanceThm}
Here we provide proof of Theorem \ref{thm:SmallToLargeConcordanceThm}, the statement of which is repeated immediately below:
\medskip

\noindent
\textbf{Theorem \ref{thm:SmallToLargeConcordanceThm}} \emph{A normal network is concordant if its fully open extension is concordant. In particular, a weakly reversible network is concordant if its fully open extension is concordant.}

\begin{proof} We suppose that the fully open extension of a normal network \rnet\\ is concordant but that \rnet\  itself is not concordant. This will produce a contradiction. As usual we denote by $S \subset \RS$ the stoichiometric subspace for the base network \rnet. Note that the stoichiometric subspace for the fully open extension coincides with \RS.

	If the network \rnet is not concordant then, for the network, there is a nonzero $\sigma \in S$ and an $\alpha \in \ker L$ which together satisfy the conditions in Definition \ref{def:Concordant}. In particular, we have
\begin{equation}
\label{eq:alphaeqn2}
\sum_{\rxn \in \scrR}\alpha_{\rxn}(y' - y) = 0.
\end{equation}
Given such a $\sigma, \alpha$ pair, our aim will be to show that the fully open extension $\{\scrS, \tilde{\scrC}, \tilde{\scrR} \}$ cannot be concordant, which contradicts the hypothesis. That is, we will want to show the existence of a nonzero $\tilde{\sigma} \in \RS$ and an $\tilde{\alpha} \in \mathbb{R}^{\tilde{\mathscr{R}}}$ satisfying 
\begin{equation}
\label{eq:alphaTildeEqn}
\sum_{\rxn \in \tilde{\scrR}}\tilde{\alpha}_{\rxn}(y' - y)  = 0
\end{equation}
such that the $\tilde{\sigma}, \tilde{\alpha}$ pair also satisfies the requirements in Definition \ref{def:Concordant}.

	For the normal network \rnet, let $T:S \to S$ be as in Definition \ref{def:normal}. Because $T$ is nonsingular, the map $[T - \epsilon I]^{-1}: S \to S$ exists for sufficiently small $\epsilon$. In fact, in a sufficiently small neighborhood of $0$ the function $\epsilon \to [T - \epsilon I]^{-1}$ is continuous, where the set of linear maps from $S$ to $S$ is given the usual norm topology. (See \cite{Kato1982}, \S 5.2.)

	For small $\epsilon \geq 0$, let $\chi(\epsilon) := [T - \epsilon I]^{-1}\sigma$, which is to say that 
\begin{equation}
T\chi(\epsilon) = \sigma + \epsilon \chi(\epsilon).
\end{equation}
(Note that $\chi(0) = T^{-1}\sigma$.)  Thus, we can write

\begin{equation}
\label{eq:chisigmaeqn}
 \sum_{y \to y' \in \scrR}\eta_{y \to y'}(y * \chi(\epsilon))(y' - y) + \sum_{s \in \scrS}(\sigma_s + \epsilon \chi_s(\epsilon))(-s) = 0.
\end{equation}

	In anticipation of construction of $\tilde{\alpha}$, we multiply \eqref{eq:alphaeqn2} by a positive number $M$, to be determined later, and add the result to \eqref{eq:chisigmaeqn} to obtain 

\begin{equation}
\label{eq:chisigmaeqn2}
 \sum_{y \to y' \in \scrR}\{M \alpha_{y \to y'} + \eta_{y \to y'}(y * \chi(\epsilon))\}(y' - y) + \sum_{s \in \scrS}(\sigma_s + \epsilon \chi_s(\epsilon))(-s) = 0.
\end{equation}
It will be useful to keep in mind that $-s$ is the reaction vector corresponding to the reaction $s \to 0 \in \tilde{\scrR}$.

	Hereafter, we choose

\begin{equation}
\label{eq:SigmaTildeChoice}
\tilde{\sigma} := \sigma + \epsilon \chi(\epsilon),
\end{equation}
where it is understood that $\epsilon > 0$ is chosen sufficiently small as to satisfy the following requirement:
\begin{equation}
\label{eq:epsilonrequirement}
\sgn (\sigma + \epsilon \chi(\epsilon))_s = \sgn \sigma_s \mbox{ when } \sgn \sigma_s \neq 0.
\end{equation}
Because $\sigma \neq 0$, \eqref{eq:epsilonrequirement} ensures that $\tilde{\sigma} \neq 0$.

	We turn now to a choice of $\tilde{\alpha}$. By way of preparation, note that we have not precluded the possibility that \scrR \;might  contain \mbox{species-degradation} reactions of the form $s \to 0$, where $s$ is a species. In fact, let $\scrM \subset \scrS$ be the set of all species for which there are such reactions; that is, $\scrM = \{s \in \scrS: s \to 0 \in \scrR\}$. We denote by $\scrM '$ the complement of \scrM \;in \scrS. By $\scrR '$ we mean the set of all reactions in \scrR \;that are not \mbox{species-degradation} reactions; that is,  $\scrR '$ is the set of all members of \scrR \  not of the form $s \to 0, \; s \in \scrM$.

	With this as background, we choose $\tilde{\alpha} \in \mathbb{R}^{\tilde{\mathscr{R}}}$ as follows:
\begin{align}
\label{EQ:AlphaTildeChoice}
&\tilde{\alpha}_{y \to y'} := M \alpha_{y \to y'} + \eta_{y \to y'} (y * \chi(\epsilon)) \;\mbox{ for all } y \to y' \in \mathscr{R}', \\
\nonumber &\tilde{\alpha}_{s \to 0} := M \alpha_{s \to 0} + \eta_{s \to 0} (s * \chi(\epsilon)) + (\sigma_s + \epsilon \chi_s(\epsilon))\; \mbox{  for all } s \in \mathscr{M}, \\
\nonumber &\tilde{\alpha}_{s \to 0}:= (\sigma_s + \epsilon \chi_s(\epsilon))\; \mbox{ for all } s \in \scrM ', 
\end{align}
and we suppose that $M > 0$ is sufficiently large as to satisfy the following requirement: For all $\rxn \in \scrR$
\begin{equation}
\label{eq:MRequirement}
\sgn \tilde{\alpha}_{\rxn} = \sgn \alpha_{\rxn} \;\;  \mathrm{if} \;\; \sgn \alpha_{\rxn} \neq 0.
\end{equation}
Note that, by virtue of \eqref{eq:chisigmaeqn2}, $\tilde{\alpha}$ so chosen is a solution of \eqref{eq:alphaTildeEqn}. 

	It remains to be shown that this $\tilde{\sigma}, \tilde{\alpha}$ pair, constructed as indicated, satisfies the conditions in Definition \ref{def:Concordant}. There are several cases to consider, corresponding to various categories of reactions in $\tilde{\scrR}$:
\medskip

\noindent
1. For  reactions of the form $s \to 0, s \in \scrM '$, we have $\sgn \tilde{\alpha}_{s \to 0} = \sgn \tilde{\sigma}_s$ so, for them, the conditions in Definition \ref{def:Concordant} are clearly satisfied.

\medskip
\noindent
2. For  reactions of the form $s \to 0, s \in \scrM$, there are two possibilities:
\medskip

(a) If $\alpha_{s \to 0} \neq 0$ then we have $\sgn \tilde{\alpha}_{s \to 0} = \sgn \alpha_{s \to 0} = \sgn \sigma_s = \sgn \tilde{\sigma}_s$ so Condition (i) in Definition \ref{def:Concordant} is satisfied.

\medskip
(b) If $\alpha_{s \to 0} = 0$ then $\sigma_s = 0$ and $\sgn \tilde{\alpha}_{s \to 0} = \sgn (\eta_{s \to 0}(s * \chi(\epsilon)) + \epsilon \chi_s(\epsilon)) = \sgn \chi_s(\epsilon) = \sgn(\sigma_s + \epsilon \chi_s(\epsilon)) = \sgn \tilde{\sigma}_s$. 

Therefore, either Condition (i) or Condition (ii) in Definition \ref{def:Concordant} is satisfied.
\medskip

\noindent
3. For $\rxn \in \scrR '$ such that $\alpha_{\rxn} \neq 0$ there must exist $s \in \supp y$ such that $0 \neq \sgn \tilde{\alpha}_{y \to y'} = \sgn \alpha_{y \to y'} = \sgn \sigma_s = \sgn \tilde{\sigma}_s$, so that Condition (i) in Definition \ref{def:Concordant} is satisfied.

\medskip 
\noindent
4. If  $\rxn \in \scrR '$ is such that $\alpha_{y \to y'} = 0$ and $\supp y \cap \supp \sigma = \emptyset$ then
\begin{equation}
\label{eq:y*chi}
\sgn \tilde{\alpha}_{y \to y'} = \sgn (y * \chi(\epsilon)),
\end{equation}
and, for all $s \in \supp y$,
\begin{equation}
\sgn \tilde{\sigma}_s = \sgn (\sigma_s + \epsilon \chi_s(\epsilon)) = \sgn \chi_s(\epsilon).
\end{equation}
There are two cases to consider:

\medskip
(a)\; If $y * \chi(\epsilon) \neq 0$ then, from \eqref{eq:y*chi}, we have the existence of $s' \in \supp y$ such that $\sgn \tilde{\alpha}_{\rxn} = \sgn y * \chi(\epsilon) = \sgn \chi_{s'}(\epsilon) = \sgn(\sigma_{s'} + \epsilon\chi_{s'}(\epsilon)) = \sgn \tilde{\sigma}_{s'}$. Thus, Condition (i) in Definition \ref{def:Concordant} is satisfied.

\medskip
(b)\; If $y*\chi(\epsilon) = 0$ then either $\supp y\  \cap \ \supp \chi(\epsilon) = \emptyset$, in which case \mbox{$\supp y \  \cap\  \supp \tilde{\sigma}  = \emptyset$}, or there exist $s, s' \in \supp y$ such that $\sgn \chi_s(\epsilon) = -\sgn \chi_{s'}(\epsilon)$, both nonzero, in which case we have $\sgn \tilde{\sigma}_s = - \sgn \tilde{\sigma}_{s'}$, both nonzero. In either case, then, Condition (ii) in Definition \ref{def:Concordant} is satisfied.

\medskip
\noindent
5. Again consider  $\rxn \in \scrR '$ such that $\alpha_{y \to y'} = 0$, but now suppose and there exist $s, s' \in \supp y$ such that $\sgn \sigma_s = - \sgn \sigma_{s'}$, both non-zero. Then from \eqref{eq:epsilonrequirement} we have $\sgn \tilde{\sigma}_s = -\sgn \tilde{\sigma}_{s'}$, both non-zero. 	Regardless of the sign of $\tilde{\alpha}_{\rxn}$, therefore, the conditions in Definition \ref{def:Concordant} are satisfied.

	In summary, then, the $\tilde{\sigma}, \tilde{\alpha}$ pair, constructed as indicated, ensures that  $\{\scrS, \tilde{\scrC}, \tilde{\scrR} \}$, the fully open extension of the (presumed discordant) normal network \rnet \, is not concordant. This, however, contradicts what has been supposed about the fully open extension. Therefore, \rnet is concordant.
\end{proof}

	We repeat the statement of Theorem \ref{thm:SmallToLargeStrongConcordanceThm} below:
\medskip

\noindent
\textbf{Theorem \ref{thm:SmallToLargeStrongConcordanceThm}} \emph{ A normal network is strongly concordant if its fully open extension is strongly concordant. In particular, a weakly reversible network is strongly concordant if its fully open extension is strongly concordant.}

\medskip

	Proof of the theorem is very similar to the proof of Theorem \ref{thm:SmallToLargeConcordanceThm}. The constructions of $\tilde{\alpha}$ and $\tilde{\sigma}$ are identical. The only substantive changes are in items 3 and 5:

\smallskip

\noindent 3. For $\rxn \in \scrR '$ such that $\alpha_{\rxn} \neq 0$ there must exist a species $s$ such that $0 \neq (\sgn \tilde{\alpha}_{y \to y'})(\sgn (y - y')_s) = (\sgn \alpha_{y \to y'})(\sgn (y - y')_s) = \sgn \sigma_s = \sgn \tilde{\sigma}_s$, so that Conditions (i) and (ii) in Definition \ref{DEF:StronglyConcordant} are satisfied. 
\smallskip

\noindent 5. Consider  $\rxn \in \scrR '$ such that $\alpha_{y \to y'} = 0$, but now suppose and there exist species $s, s'$ such that
\begin{equation}
 \sgn \sigma_{s} = \sgn (y - y')_s \neq 0 \nonumber
\end{equation} 
and
\begin{equation}
\sgn \sigma_{s'} = - \sgn (y - y')_{s'} \neq 0. \nonumber
\end{equation} 
From  \eqref{eq:epsilonrequirement} we also have $\sgn \tilde{\sigma}_s = \sgn \sigma_s$ and $\sgn \tilde{\sigma}_{s'} = \sgn \sigma_{s'}$. Regardless of the sign of $\tilde{\alpha}_{\rxn}$, then, the conditions in Definition \ref{DEF:StronglyConcordant} are satisfied. 

\section{Proofs of Theorems \ref{thm:EigenvalueThm} and \ref{thm:DiscordanceEigenvalueThm}}
\label{App:EigenvalueProofs}

	We repeat below the statement of Theorem \ref{thm:EigenvalueThm}:

\medskip
\noindent
\textbf{Theorem \ref{thm:EigenvalueThm}} \emph{Let \kinsys \  be a kinetic system with stoichiometric subspace $S$ and species formation rate function $f: \PbarS \to S$. Moreover, suppose that the kinetics is differentiably monotonic at $c^* \in \PS$. If the underlying network \rnet is concordant then the derivative $df(c^*): S \to S$ is nonsingular (whereupon 0 is not one of its eigenvalues). If the network's fully open extension is concordant then no real eigenvalue of  $df(c^*)$ is positive.}

\begin{proof} For the sake of brevity, for each $\rxn \in \scrR$ we denote by $p_{\rxn} \in \PbarS$ the gradient of $\scrK_{\rxn}(\cdot)$ at $c^*$. That is, for each $s \in \scrS$,
\begin{equation}
p^s _{\rxn} := \frac{\partial \scrK_{\rxn}}{\partial c_s}(c^*)
\end{equation}
For each  $\rxn \in \scrR$, note that $\supp p_{\rxn} = \supp y$. In terms of the notation we have now introduced,  $df(c^*)$ takes the following form: For each $\sigma \in S$,
\begin{equation}
 df(c^*)\sigma = \sum_{\rxn \in \scrR}(p_{\rxn} \cdot \sigma)(y' - y).
\end{equation}

	Suppose that \rnet \ is concordant. Suppose also, contrary to what is to be proved, that $df(c^*)$ is singular. Then there is a nonzero $\sigma \in S$ such that 
\begin{equation}
\sum_{\rxn \in \scrR}(p_{\rxn} \cdot \sigma)(y' - y) = 0.
\end{equation}
If we take $\alpha \in \RR$ to be defined by $\alpha_{\rxn} :=  p_{\rxn} \cdot \sigma$, then it is not difficult to see that, taken together, $\sigma$ and $\alpha$ satisfy the conditions of Definition \ref{def:Concordant}. Thus, \rnet is discordant, and we have a contradiction.

	Now suppose that the fully open extension of \rnet\  is concordant and that, contrary to what is to be proved,  $df(c^*)$  has a positive real eigenvalue. Then there is a nonzero  $\tilde{\sigma} \in S$ and a positive real number $\lambda$ such that  $df(c^*)\tilde{\sigma} = \lambda\tilde{\sigma}$, which can be written in the following way:
\begin{equation}
\label{eq:EigenvalEqn1}
\sum_{\rxn \in \scrR}(p_{\rxn} \cdot \tilde{\sigma})(y' - y) + \sum_{s \in \scrS}\lambda \tilde{\sigma}_s(-s) = 0.
\end{equation}
Note that $-s$ is the reaction vector for reaction $s \to 0$ in the fully open extension  $\{\scrS, \tilde{\scrC}, \tilde{\scrR} \}$ of the network \rnet. (Recall the notation given in \eqref{eq:FullyOpenExtNotation}.)

	It is possible that \scrR,  the reaction set for the original network, contains species-degradation reactions of the form $s \to 0$, where $s$ is a species. Let $\scrM \subset \scrS$ be the set of all species for which there are such reactions; that is, $\scrM = \{s \in \scrS: s \to 0 \in \scrR\}$. We denote by $\scrM '$ the complement of \scrM \;in \scrS. By $\scrR '$ we mean the set of all reactions in \scrR \;that are not \mbox{species-degradation} reactions; that is,  $\scrR '$ is the set of all members of \scrR \ that are not of the form $s \to 0, \; s \in \scrM $. Then \eqref{eq:EigenvalEqn1} takes the form
\begin{equation}
\label{eq:EigenvalEqn2}
\sum_{\rxn \in \scrR'}(p_{\rxn} \cdot \tilde{\sigma})(y' - y) + \sum_{s\, \in\, \scrM}(p_{s \to 0} + \lambda) \tilde{\sigma}_s(-s) + \sum_{s\, \in\, \scrM '}\lambda \tilde{\sigma}_s(-s) = 0.
\end{equation}

	Now we choose  $\tilde{\alpha} \in \mathbb{R}^{\tilde{\mathscr{R}}}$ to be defined by $\tilde{\alpha}_{\rxn} :=  p_{\rxn} \cdot \tilde{\sigma}$ for $\rxn \ \scrR'$, $\tilde{\alpha}_{s \to 0} := (p_{s \to 0} + \lambda) \tilde{\sigma}_s$ for $s \in \scrM$, and $\tilde{\alpha}_{s \to 0} :=  \lambda \tilde{\sigma}_s$ for $s \in \scrM '$. Thus, $\tilde{\alpha}$ satisfies
\begin{equation}
\sum_{\rxn \in \tilde{\scrR}}\tilde{\alpha}_{\rxn}(y' - y) = 0. \nonumber
\end{equation} 
It is easy to see that, taken together, $\tilde{\sigma}$ and $\tilde{\alpha}$ satisfy the conditions of Definition \ref{def:Concordant}. Thus, the fully open extension  $\{\scrS, \tilde{\scrC}, \tilde{\scrR} \}$ is discordant, and we again have a contradiction. It follows, then, that no real eigenvalue of  $df(c^*)$ is positive.
\end{proof}
\medskip

	We repeat below the statement of Theorem \ref{thm:DiscordanceEigenvalueThm}:

\bigskip
\noindent
\textbf{Theorem \ref{thm:DiscordanceEigenvalueThm}}  \emph{ Consider a reaction network with positively dependent reaction vectors. If the network is discordant, then there exists for it a differentiably monotonic kinetics such that the resulting kinetic system admits a positive \emph{degenerate} equilibrium. If the network's fully open extension is discordant then there exists for the original network a 
differentiably monotonic kinetics such that the resulting kinetic system 
admits an unstable positive equilibrium --- in fact, a positive equilibrium associated with a positive real eigenvalue.}

\begin{proof}
Throughout this proof we consider a reaction network \rnet \ with stoichiometric subspace $S$, and we suppose that the reaction vectors are positively dependent. Thus, we have positive numbers $\{\kappa_{\rxn}\}_{\rxn \in \scrR}$ that satisfy
\begin{equation}
\label{eq:KappaEquation}
\sum_{\rxn \in \scrR}\kappa_{\rxn}(y' - y) = 0.
\end{equation}

	We divide the proof into two parts:
\medskip

(i) First we suppose that the network \rnet \ is itself discordant. Then there exists  a pair, $\alpha \in \RR$ and nonzero $\sigma \in S$, that together satisfy conditions (i) and (ii) in Definition \ref{def:Concordant}. From those conditions it is not difficult to see that, for each reaction $\rxn \in \scrR$ there is a vector $p_{\rxn} \in \PbarS$ such that
\begin{equation}
\supp p_{\rxn} =  \supp y \quad \mathrm{and} \quad \alpha_{\rxn} = p_{\rxn} \cdot \sigma.
\end{equation}
\noindent Thus we can write
\begin{equation}
\label{eq:FirstpEqn}
\sum_{\rxn \in \scrR}(p_{\rxn} \cdot \sigma)(y' - y) = 0.
\end{equation}
Let $c^*$ be any member of \PS, and for each $\rxn \in \scrR$ let $q_{\rxn} \in \PbarS$ be defined by
\begin{equation}
\label{eq:FirstqDef}
q_{\rxn} := \frac{c^* \circ  p_{\rxn}}{\kappa_{\rxn}}.
\end{equation}
Furthermore let $k \in \PR$ be defined by the requirement that
\begin{equation}
\label{eq:FirstkDef}
k_{\rxn}(c^*)^{q_{\rxn}} = \kappa_{\rxn}, \quad \forall \rxn \in \scrR.
\end{equation}
We construct a differentiably monotonic kinetics for the network as follows: For each $\rxn \in \scrR$
\begin{equation}
\scrK_{\rxn}(c) := k_{\rxn}c^{q_{\rxn}}, \quad \forall c \in \PbarS.
\end{equation}
With this kinetics the species-formation-rate function takes the form
\begin{equation}
f(c) := \sum_{\rxn \in \scrR}k_{\rxn}c^{q_{\rxn}}(y' - y).
\end{equation}
	From \eqref{eq:KappaEquation} and \eqref{eq:FirstkDef} it follows that, for the kinetic system \kinsys, $c^*$ is a positive equilibrium. Taken with \eqref{eq:FirstkDef} and \eqref{eq:FirstqDef}, some computation gives
\begin{equation}
df(c^*)\sigma = \sum_{\rxn \in \scrR}(p_{\rxn} \cdot \sigma)(y' - y).
\end{equation}
Because $\sigma$ is not zero, it follows from \eqref{eq:FirstpEqn} that $df(c^*)$ is singular, whereupon $c^*$ is a degenerate positive equilibrium.\medskip

(ii) Next we suppose instead that the fully open extension of \rnet \  is discordant. Then there are numbers $\{\alpha_{\rxn}\}_{\rxn \in \scrR}$ and $\{\alpha_{s \to 0}\}_{s \in \scrS}$ satisfying
\begin{equation}
\label{eq:ExtendedAlphEqn}
\sum_{\rxn \in \scrR}\alpha_{\rxn}(y' - y) + \sum_{s \in \scrS}\alpha_{s \to 0}(-s) = 0
\end{equation}
and also a nonzero $\sigma$ such that conditions (i) and (ii) of Definition \ref{def:Concordant} are satisfied\footnote{We have not precluded the possibility that a degradation reaction, say $\bar{s} \to 0$, might be a member of \scrR, the set of reactions in the original network.  In this case the term $\alpha_{\bar{s} \to 0}(-\bar{s})$ would appear twice in \eqref{eq:ExtendedAlphEqn}, once in the first sum and once in the second. Nevertheless, it is easy to confirm that the sentence remains true as written.}. Again we can choose for each reaction $\rxn \in \scrR$ a vector $p_{\rxn} \in \PbarS$ such that
\begin{equation}
\supp p_{\rxn} =  \supp y \quad \mathrm{and} \quad \alpha_{\rxn} = p_{\rxn} \cdot \sigma.
\end{equation}
Moreover, we can choose positive numbers $\{m_s\}_{s \in \scrS}$ such that
\begin{equation}
\alpha_{s \to 0} = m_s\sigma_s.
\end{equation}
Thus, \eqref{eq:ExtendedAlphEqn} can be rewritten

\begin{equation}
\label{eq:ExtendedSigmaEqn}
\sum_{\rxn \in \scrR} (p_{\rxn} \cdot \sigma)(y' - y) = \sum_{s \in \scrS}m_s \sigma_s(s).
\end{equation}

	Now let 
\begin{equation}
\bar{\sigma} := \sum_{s \in \scrS}m_s\sigma_s(s).
\end{equation}
From \eqref{eq:ExtendedSigmaEqn} if follows that $\bar{\sigma}$ is a member of $S$, the stoichiometric subspace for the original network \rnet. Moreover, because $\sigma \neq 0$, it is easy to see that $\sigma \cdot \bar{\sigma} > 0$, whereupon $\bar{\sigma} \neq 0$.

	If, for each reaction $\rxn \in \scrR$, we let $\bar{p}_{\rxn} \in \PbarS$ be defined by
\begin{equation}
\bar{p}^s_{\rxn} := \frac{p^s_{\rxn}}{m_s}, \quad \forall s \in \scrS,
\end{equation}
then \eqref{eq:ExtendedSigmaEqn} can be rewritten as
\begin{equation}
\label{eq:SecondExtendedSigmaEqn}
\sum_{\rxn \in \scrR} (\bar{p}_{\rxn} \cdot \bar{\sigma})(y' - y) = \bar{\sigma}.
\end{equation}

	Hereafter, this part of the proof will very much resemble part (i):
Let $c^*$ be any member of \PS, and for each $\rxn \in \scrR$\  let $q_{\rxn} \in \PbarS$ be defined by
\begin{equation}
q_{\rxn} := \frac{c^* \circ  \bar{p}_{\rxn}}{\kappa_{\rxn}}.
\end{equation}
As before, we let $k \in \PR$ be defined by the requirement that
\begin{equation}
k_{\rxn}(c^*)^{q_{\rxn}} = \kappa_{\rxn}, \quad \forall \rxn \in \scrR,
\end{equation}and also as before we construct a differentiably monotonic kinetics for the network \rnet \ in the following way: For each $\rxn \in \scrR$
\begin{equation}
\scrK_{\rxn}(c) := k_{\rxn}c^{q_{\rxn}}, \quad \forall c \in \PbarS.
\end{equation}
Again, the species-formation-rate function takes the form
\begin{equation}
f(c) := \sum_{\rxn \in \scrR}k_{\rxn}c^{q_{\rxn}}(y' - y).
\end{equation}
Note that $c^*$ is an equilibrium of the kinetic systems \kinsys, and some computation gives
\begin{equation}
df(c^*)\bar{\sigma} = \bar{\sigma}.
\end{equation}
Thus, $\bar{\sigma}$ is an eigenvector of $df(c^*)$ corresponding to the positive real eigenvalue $1$. Thus, the equilibrium $c^*$ is unstable.
\end{proof}

\bibliographystyle{amsplain}
\bibliography{MonoLibrary}

\end{document}